\documentclass[lettersize,journal]{IEEEtran}
\usepackage{cite}
\usepackage{amsmath,amssymb,amsfonts,amsthm}
\usepackage{comment}

\usepackage{algorithmic}
\usepackage{algorithm}
\usepackage{enumerate}  
\DeclareMathOperator{\Tr}{Tr}
\DeclareMathOperator{\diag}{diag}

\DeclareMathOperator*{\argmax}{arg\,max} 

\usepackage[english]{babel}
\usepackage{hyperref}

\newtheorem{remark}{Remark}
\newtheorem{proposition}{Proposition}

\usepackage{subcaption}
\usepackage{graphicx}
\usepackage{textcomp}
\usepackage{optidef}
\usepackage{xcolor}
\def\BibTeX{{\rm B\kern-.05em{\sc i\kern-.025em b}\kern-.08em
    T\kern-.1667em\lower.7ex\hbox{E}\kern-.125emX}}
\usepackage{balance}

\begin{document}

\title{Toward Efficient and Privacy-Aware eHealth Systems: An Integrated Sensing, Computing, and Semantic Communication Approach}

\author{Yinchao Yang, Yahao Ding, Zhaohui Yang, Chongwen Huang, Zhaoyang Zhang, Dusit Niyato, \IEEEmembership{Fellow, IEEE}, and Mohammad Shikh-Bahaei,
\IEEEmembership{Senior Member, IEEE}
\thanks{Yinchao Yang, Yahao Ding, and Mohammad Shikh-Bahaei are with the Department of Engineering, King's College London, London, UK. (emails: yinchao.yang@kcl.ac.uk; yahao.ding@kcl.ac.uk; m.sbahaei@kcl.ac.uk).}
\thanks{Dusit Niyato is with the College of Computing and Data Science, Nanyang Technological University, Singapore 639798, Singapore (email: dniyato@ntu.edu.sg).}
\thanks{Zhaohui Yang, Chongwen Huang and Zhangyang Zhang are with the College of Information Science and Electronic Engineering, Zhejiang University, Hangzhou, Zhejiang 310027, China, and Zhejiang Provincial Key Lab of Information Processing, Communication and Networking (IPCAN), Hangzhou, Zhejiang, 310007, China. (emails: yang\_zhaohui@zju.edu.cn; chongwenhuang@zju.edu.cn; ning\_ming@zju.edu.cn).}
}

\markboth{Journal of \LaTeX\ Class Files,~Vol.~18, No.~9, September~2020}%
{How to Use the IEEEtran \LaTeX \ Templates}
\maketitle

\begin{abstract} 
Real-time and contactless monitoring of vital signs, such as respiration and heartbeat, alongside reliable communication, is essential for modern healthcare systems, especially in remote and privacy-sensitive environments. Traditional wireless communication and sensing networks fall short in meeting all the stringent demands of eHealth, including accurate sensing, high data efficiency, and privacy preservation. To overcome the challenges, we propose a novel integrated sensing, computing, and semantic communication (ISCSC) framework. In the proposed system, a service robot utilises radar to detect patient positions and monitor their vital signs, while sending updates to the medical devices. Instead of transmitting raw physiological information, the robot computes and communicates semantically extracted health features to medical devices. This semantic processing improves data throughput and preserves the clinical relevance of the messages, while enhancing data privacy by avoiding the transmission of sensitive data. Leveraging the estimated patient locations, the robot employs an interacting multiple model (IMM) filter to actively track patient motion, thereby enabling robust beam steering for continuous and reliable monitoring. We then propose a joint optimisation of the beamforming matrices and the semantic extraction ratio, subject to computing capability and power budget constraints, with the objective of maximising both the semantic secrecy rate and sensing accuracy. Simulation results validate that the ISCSC framework achieves superior sensing accuracy, improved semantic transmission efficiency, and enhanced privacy preservation compared to conventional joint sensing and communication methods. 
\end{abstract}

\begin{IEEEkeywords}
Integrated sensing and communication, semantic communication, and vital sign detection.
\end{IEEEkeywords}

\IEEEpeerreviewmaketitle

\section{Introduction}

\IEEEPARstart{C}{ontinuous} and real-time monitoring of vital signs, such as respiration and heartbeat, plays a pivotal role in the early detection and prevention of potentially life-threatening conditions \cite{li2024passive}. Timely identification of abnormalities like irregular heartbeats is essential, as they often serve as early indicators of cardiovascular disorders \cite{zhai2022contactless}. Detecting these irregularities at an early stage enables prompt medical intervention, reducing the risk of severe complications and improving patient outcomes.

Vital sign monitoring methods can be broadly classified into two categories: contact-based and contactless. Traditional contact-based approaches, such as smartwatches and multi-parameter monitors, require direct physical contact with the body, which can cause discomfort and inconvenience for patients \cite{zhang2024single}.
In contrast, contactless vital sign detection has gained considerable attention due to its ability to provide non-intrusive monitoring while preserving patient privacy \cite{wang2021mmhrv}. These systems utilise wireless signals, such as millimetre-wave (mmWave), to detect subtle physiological movements, including respiration, heartbeat, and other body motions, by analysing signal reflections from the human body \cite{liu2024diversity, fioranelli2021contactless}.
The integration of contactless monitoring into eHealth systems offers significant advantages for remote healthcare, telemedicine, and personalised health tracking, positioning contactless sensing as a key driver in the evolution of next-generation intelligent healthcare solutions \cite{manimegalai2023exploring}.

Seamless data exchange, both within and between healthcare facilities, is also fundamental to the effectiveness of eHealth systems. Specifically, the reliable transmission of data during continuous sensing operations is paramount for delivering effective and responsive healthcare \cite{dolas2024integrated}. Integrated sensing and communication (ISAC) presents a promising paradigm by unifying wireless communication and contactless sensing functionalities within a single system \cite{wen2024survey, meng2024cooperative, zhu2024enabling,qi2024joint}. This dual capability allows systems to concurrently support both communication and sensing tasks, thereby enhancing system efficiency and spectrum utilisation \cite{liu2022integrated,qi2022integrating,qi2020integrated, olfat2008optimum, bobarshad2010low, nehra2010cross, nehra2010spectral, jia2020channel, kobravi2007cross}.

While ISAC optimises spectrum utilisation by enabling dual functionality in sensing and communication, next-generation healthcare systems demand more intelligent and efficient data transmission mechanisms \cite{pathak2022sembox,long2025survey,guo2025survey}. Conventional wireless communication frameworks typically emphasise the reliable delivery of raw data, which can lead to excessive bandwidth consumption and the transmission of redundant or non-informative content. In contrast, modern eHealth applications require semantically aware communication strategies that prioritise clinically meaningful information, thereby enabling timely and accurate medical decision-making \cite{dolas2024integrated}. Semantic communication addresses this need by shifting the focus from bit-level accuracy to the transmission of relevant semantic content \cite{yang2023energy,yang2023secure,wang2025generative}. By leveraging knowledge bases (KBs), such as electronic health records (EHRs) which contain patient medical records, semantic communication enables medical devices to convey high-level meanings instead of transmitting raw physiology or clinical data \cite{11096081}. For example, instead of transmitting the exact message ``Patient heart rate increased from 100 BPM to 150 BPM,'' a semantically aware system may send ``Patient heart rate increased sharply,'' thereby reducing data overhead while preserving essential clinical meaning. The exact data are stored in EHRs and shared across medical devices and, potentially, interconnected healthcare facilities. This provides a complete view of a patient’s health over time, helping to create a more efficient and responsive healthcare system \cite{rahman2024next}. Semantic communication has attracted considerable research attention in recent years. For instance, the authors in \cite{peng2024robust} proposed a machine learning-based semantic text communication system that is robust against physical channel noise and literal semantic impairments, such as misinterpretation of contextual meaning and word ambiguity. Their results demonstrated enhanced interpolation at the receiver with reduced inference time. In \cite{evgenidis2024hybrid}, a hybrid semantic-Shannon communication system was introduced, allowing each subcarrier to dynamically select between semantic and Shannon communication modes to optimise both delay and accuracy requirements. Furthermore, the authors in \cite{hu2024semantic} explored the use of relays as edge servers to deliver semantic communication services for both semantic users with abundant computational resources and conventional users with limited capabilities. Their results showed improved communication performance while minimising the use of resources such as bandwidth and transmit power. Nevertheless, despite these notable advancements, current research predominantly targets semantic communication enhancements in general wireless networks, while the exploration of semantic communication applications within the eHealth domain remains relatively limited \cite{madhavi2025iot, karahan2025towards, yuan2024universal}.

Although ISAC and semantic communication each demonstrate considerable potential individually, there remains a critical gap in research exploring their integration for eHealth. In an eHealth setting, joint sensing and semantic communication can enhance both the accuracy of patient monitoring and the efficiency of medical data transmission \cite{11096081}. For instance, radar sensors can detect subtle physiological movements \cite{wang2023vital}, while semantic processing can intelligently extract and transmit only the most clinically relevant data \cite{athama2024semantic}, minimising network congestion and enabling faster medical decision-making. Another significant advantage of integrating ISAC and semantic communication is the enhancement of privacy and security for patient data.  Privacy and security are especially critical in healthcare systems due to the sensitive nature of medical information. Although several recent works \cite{fan2024authentic, ahmed2025privacy, mpembele2023communication, romeo2024privacy} have explored privacy-preserving mechanisms for eHealth, those works overlook the potential of semantic-level protection. Medical information, including health records, diagnostic images, and treatment plans, is highly sensitive to unauthorised access. In ISAC systems, such data may be exposed to eavesdropping \cite{manimegalai2023exploring}. While precoding techniques can enhance physical layer security, their effectiveness typically relies on non-line-of-sight (NLoS) conditions, or wiretap channel models, where the eavesdropper experiences degraded channel quality. Semantic communication further reinforces data protection by eliminating the transmission of raw physiological signals. Instead, successful decoding depends on access to shared KBs. When KBs are managed in compliance with regulations such as the general data protection regulation (GDPR), unauthorised access is effectively prevented \cite{azam2023modelling}. Consequently, even if semantic information is intercepted during transmission, eavesdroppers are unable to reconstruct meaningful content without the corresponding KBs.

To bridge the research gap concerning the integration of ISAC and semantic communication in healthcare systems, we propose a novel framework that effectively integrates these technologies to enable advanced healthcare applications. We refer to the proposed framework as \textbf{integrated sensing, computing, and semantic communication (ISCSC)}. 
Through sensing-assisted semantic communication, the ISCSC framework utilises environmental sensing to intelligently select and prioritise medical data for efficient transmission. This method not only significantly enhances data rates but also safeguards patient privacy by preventing data transmission to unintended recipients. Conversely, semantic communication-assisted sensing facilitates real-time acquisition and transmission of vital signs and medical information, enabling timely anomaly detection and further reinforcing data privacy by eliminating the need to transmit raw medical data. A comparison of ISCSC and ISAC performance in sensing and communication is illustrated in Fig.~\ref{Pareto}.

\begin{figure}
    \centering
    \includegraphics[width=\linewidth]{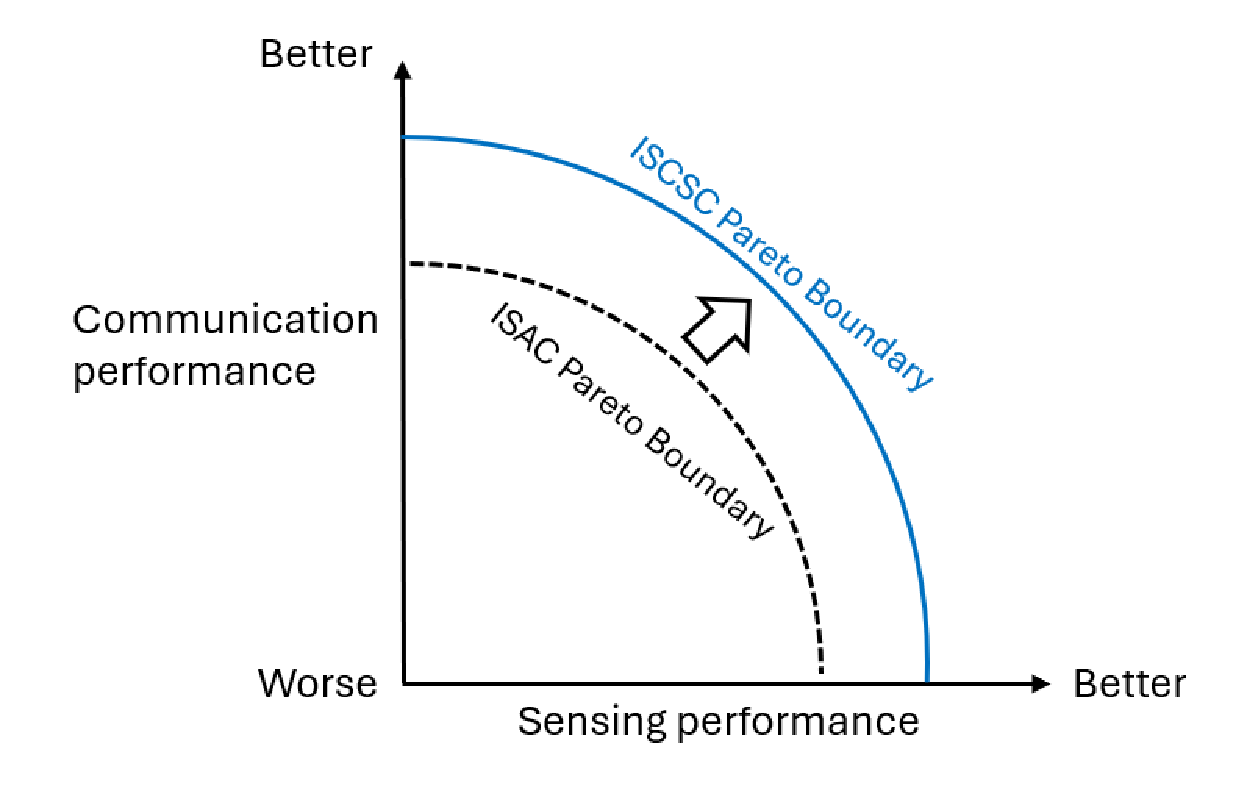}
    \caption{An illustration of the performance Pareto boundary of ISAC and ISCSC systems.}
    \label{Pareto}
\end{figure}

In summary, the key contributions of this paper are as follows:

\begin{enumerate}
    \item We propose a novel ISCSC framework specifically designed for eHealth applications, with an emphasis on vital sign detection. The proposed ISCSC design enhances data rate and privacy while preserving the same sensing accuracy as the ISAC design.

    \item We developed a privacy-aware predictive beamforming strategy to enable simultaneous downlink semantic communication to medical devices and sensing of patients. This strategy ensures that communication signals remain confidential and are not accessible to unintended recipients (e.g., patients being sensed). 
    
    \item To ensure accurate sensing in dynamic environments, the system employs an interacting multiple model (IMM) filter for real-time prediction of patient motion \cite{mazor2002interacting}. This enables dynamic beam steering to track patient positions, significantly improving the reliability and continuity of vital sign monitoring.
\end{enumerate}

The remainder of this paper is organised as follows. Section II presents the system model of the ISCSC framework. Section III details the patient motion models and the associated tracking techniques. Section IV defines the performance metrics that are used to evaluate the ISCSC system. Section V formulates the optimisation problem and introduces the proposed algorithms. Section VI provides simulation results and performance analysis, and Section VII concludes the paper. The main symbols used in this paper are summarised in Table \ref{Table of symbols}.

\begin{table}[t]
\footnotesize
    \centering
    \caption{List of main symbols.}
    \label{Table of symbols}
        \begin{tabular}{|l||l|}
        \hline
        \textbf{Symbol} & \textbf{Description}  \\
        \hline
        $K$ & Number of medical devices \\
        \hline
        $\mathcal{K}$ & Set of medical devices \\
        \hline
        $L$ & Number of patients \\
        \hline
        $\mathcal{L}$ & Set of patients \\
        \hline
        $\mathbf{x}(t)$ & Transmitted signal from the service robot \\
        \hline
        $\mathbf{c}(t)$ & Communication signal with semantic information \\
        \hline
        $\mathbf{z}(t)$ & Sensing signal \\
        \hline
        $\mathbf{w}_k$ & Transmit beamforming vector of medical device $k$ \\
        \hline
        $\mathbf{W}$ & Transmit beamforming matrix of all medical devices \\
        \hline
        $\mathbf{r}_l$ & Transmit beamforming vector of patient $l$ \\
        \hline
        $\mathbf{R}$ & Transmit beamforming matrix of all patients \\
        \hline
        $\mathbf{R}_x$ & Covariance matrix of the transmit signal \\
        \hline
        $y_k(t)$ & Received signal of medical device $k$ \\
        \hline
        $\mathbf{h}_k$ & Channel between the robot and medical device $k$ \\
        \hline
        $y_l(t)$ & Received signal of patient $l$ \\
        \hline
        $\mathbf{h}_l$ & Channel between the robot and patient $l$ \\
        \hline
        $\sigma^2_c, \sigma^2_r$ & Communication noise \\
        \hline
        $f_{r,l}, f_{h,l}$ & Respiration and heartbeat frequencies of patient $l$ \\
        \hline
        $A_{r,l}, A_{h,l}$ & Respiration and heartbeat amplitudes of patient $l$ \\
        \hline
        $m_{r,l}(t)$ & Respiration signal of patient $l$ \\
        \hline
        $m_{h,l}(t)$ & Heartbeat signal of patient $l$ \\
        \hline
        $\mathbf{\hat{y}}_l(t)$ & Echo signal \\
        \hline
        $\phi (t)$ & Instantaneous phase of the echo signal\\
        \hline
        $\mathbf{q}_{l, t}$ & State vector of a patient \\
        \hline
        $\mathbf{g}_1(\cdot), \mathbf{g}_2(\cdot)$ & State transition and Measurement functions \\
        \hline
        $\mathbf{u}_t, \mathbf{z}_t$ & State transition and Measurement noise vectors \\
        \hline
        $\rho_k$ & Semantic extraction ratio of medical device $k$ \\
        \hline
        $P_{\text{comp}}$ & Computing power required for semantic extraction \\ 
        \hline
        $P_{\text{c\&s}}$ & Signal transmission power \\ 
        \hline
        \end{tabular}
\end{table}

\subsection*{List of Notations:}
Capital boldface letters denote matrices, while lowercase boldface letters denote vectors. Scalars are represented by standard lowercase or uppercase letters. The set of complex numbers is denoted by $\mathbb{C}$, with $\mathbb{C}^{n \times 1}$ and $\mathbb{C}^{m \times n}$ representing an $n$-dimensional complex vector and an $m \times n$ complex matrix, respectively. The identity matrix is denoted by $\mathbf{I}$, and the all-zero matrix by $\mathbf{0}$. The superscripts $(\cdot)^T$ and $(\cdot)^H$ represent the transpose and Hermitian (conjugate transpose), respectively. The trace and rank of a matrix are denoted by $\Tr(\cdot)$ and $\text{rank}(\cdot)$, respectively. The expectation operator is denoted by $\mathbb{E}[\cdot]$, and $|\cdot|$ represents the absolute value or magnitude. The notation $\succeq$ indicates positive semi-definiteness. The operators $\Re\{\cdot\}$ and $\Im\{\cdot\}$ extract the real and imaginary parts, respectively. The complex Gaussian distribution with zero mean and variance $\sigma^2$ is denoted by $\mathcal{CN}(0,\sigma^2)$.

\section{System Model}

\begin{figure*}[!t]
    \centering
    \includegraphics[width=0.8\linewidth]{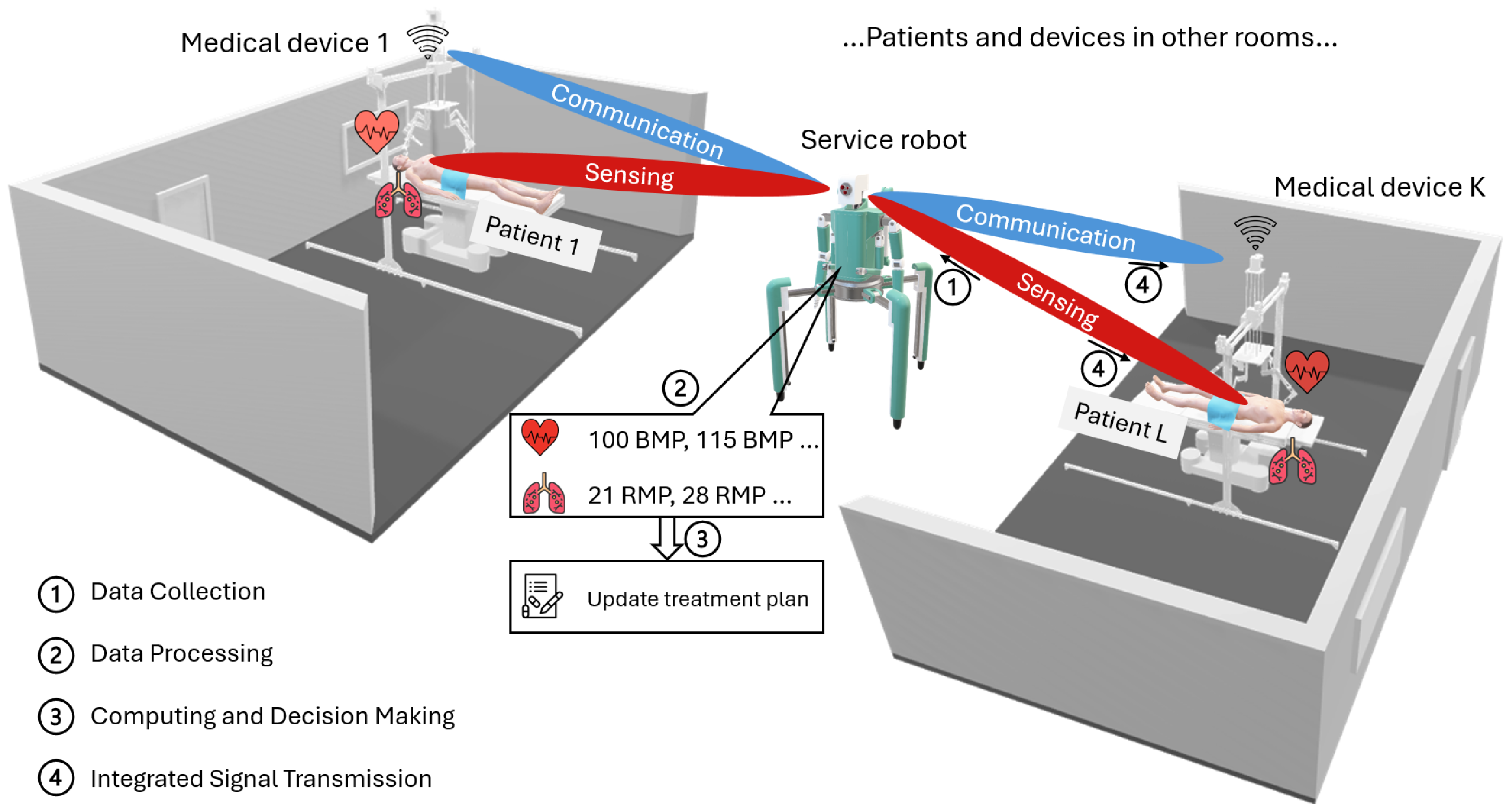}
    \caption{An ISCSC-enabled eHealth system, where a service robot engages in semantic communication with multiple medical devices and performs sensing of multiple patients to determine their locations and vital signs.}
    \label{ISCSC ehealth}
\end{figure*}

\begin{figure*}[!t]
    \centering
    \includegraphics[width=\linewidth]{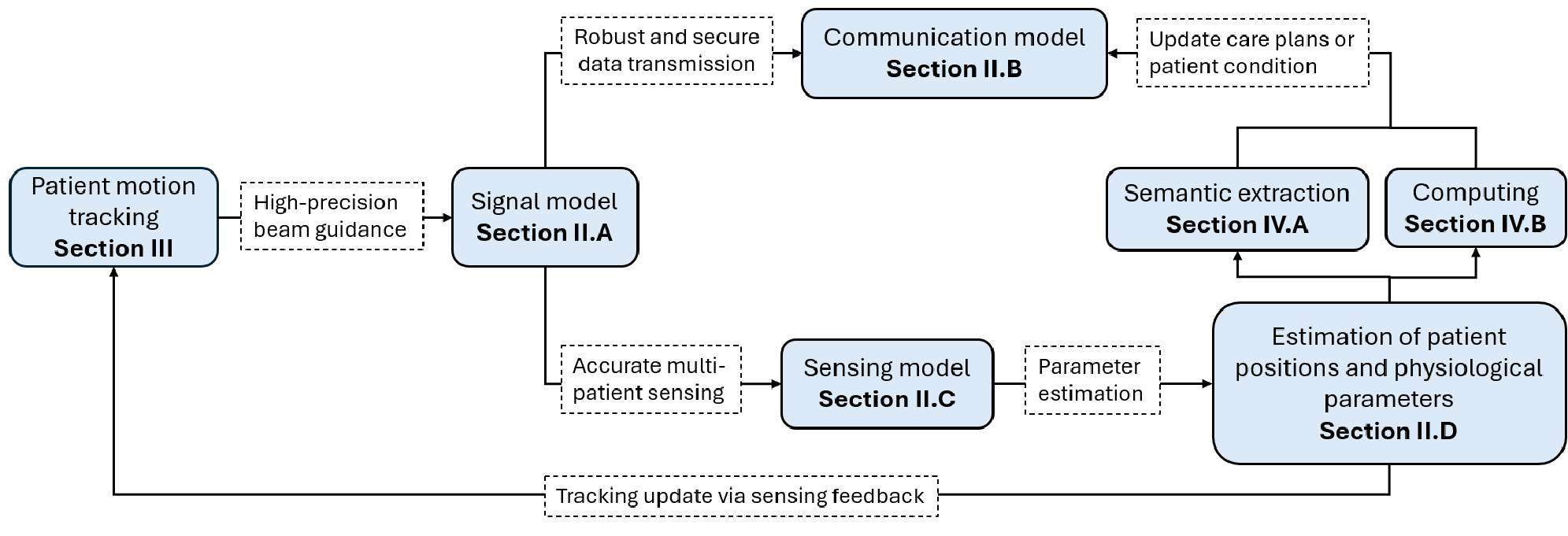}
    \caption{Flowchart of the proposed ISCSC eHealth system, illustrating the key components.}
    \label{procedure}
\end{figure*}

As illustrated in Fig.~\ref{ISCSC ehealth}, we consider the design of the ISCSC system within a hospital environment, comprising a service robot, multiple patients, and various medical devices. The service robot is equipped with a uniform linear array (ULA) consisting of \( N \) antennas and serves as a dual-functional node for both downlink communication and patient sensing. Let $\mathcal{K} = \{1, \ldots, K \}$ and $\mathcal{L} = \{1, \ldots, L \}$ denote the sets of medical devices and patients, respectively. Specifically, the robot establishes wireless communication links with \( K \) medical devices, where each device \( k \in \mathcal{K} \) is equipped with a single antenna. Concurrently, the robot performs vital sign detection for \( L \) patients, with each patient \( l \in \mathcal{L} \) modelled as a point target. During the concurrent operation of communication and sensing functionalities, ensuring the confidentiality of transmitted information is critical. Specifically, communication signals directed toward medical devices, often containing sensitive patient data, must be secured against unauthorised access, including from the individuals being sensed. The key procedures of the ISCSC eHealth system are summarised as follows:
\begin{enumerate}
    \item Data Collection: The service robot acquires sensing data from the patients using sensors like mmWave radar.
    \item Data Processing: The acquired data are processed to extract patient locations and vital physiological parameters.
    \item Computing and Decision Making: Semantic information is derived based on the processed sensing data to support treatments.
    \item Integrated Signal Transmission: A unified signal comprising both sensing and semantic information is transmitted for continuous monitoring and further communication.
\end{enumerate}

Fig.~\ref{procedure} illustrates the key modules and their benefits within the proposed system. Each component is discussed in detail in the subsequent sections.

\subsection{Signal Model}

In the considered ISCSC system, the robot simultaneously transmits both communication and sensing signals, with beamformers assisting in accurately directing the signals toward the desired direction (i.e., the chest of a patient). Precise sensing beamforming is particularly crucial for accurately analysing heart rate and respiratory rate from the sensing signal, ensuring that each patient's vital sign is properly captured. This highlights the essential role of beamformers in the system. As patients are in motion, the robot dynamically predicts their positions in each time slot and steers the beams accordingly to maintain precise tracking and optimal signal reception. Consequently, the transmitted signal from the robot can be expressed as \cite{jiang2025ris}:
\begin{equation}\label{signal model}
\mathbf{x} \left(t\right) = \mathbf{W} \mathbf{c} \left(t\right) + \mathbf{R} \mathbf{z}\left(t\right), 
\end{equation}
where $\mathbf{W} \in \mathbb{C}^{N \times K}$ and $\mathbf{R} \in \mathbb{C}^{N \times L}$ are the communication and sensing beamforming matrices, respectively. Note that the design of $\mathbf{R}$ is based on the predicted states (e.g., angle) of patients. The semantic message is denoted by $\mathbf{c} \in \mathbb{C}^{K\times1}$, and $\mathbf{z} \in \mathbb{C}^{L\times1}$ is the sensing signal. 
The semantic message $c_k \in \mathbf{c}$ is generated from the conventional communication message $m_k$, i.e., $c_k = f(m_k)$. Here, $f(\cdot)$ is a function, such as Transformers, that maps a conventional communication message to a semantic message. As an example, the original message $m_k$ ``Patient IDxx01's heart rate is 80 BPM and respiration rate is 15 RMP; Patient IDxx02's heart rate is 140 BPM and respiration rate is 30 RMP'' can be abstracted into a high-level semantic message $c_k$ ``Patient IDxx01 is stable; Patient IDxx02 needs attention,'' with higher priority assigned to the critical status of Patient IDxx02. The detailed physiological data are stored in KBs (such as the EHRs) and are accessible to authorised medical devices.

Without loss of generality, the following three assumptions are made to simplify the system design  \cite{liuxiang2020joint}: 
\begin{itemize}
    \item There is no correlation between the confidential message and the radar signal, i.e., $\mathbb{E}(\mathbf{c} \mathbf{z}^H) = \mathbf{0}_{K \times L}$.
    \item There is no correlation between the confidential messages for different medical devices, i.e., $\mathbb{E}(\mathbf{c}\mathbf{c}^H) = \mathbf{I}_{K}$.
    \item There is no correlation between the sensing signals for different patients, i.e., $\mathbb{E}(\mathbf{z}\mathbf{z}^H) = \mathbf{I}_{L}$.
\end{itemize}

As a result, the covariance matrix of the transmit waveform is given by  
\begin{equation}\label{eq9}
    \mathbf{R}_x = \mathbb{E}[\mathbf{x}\mathbf{x}^H] = \sum_{k=1}^K \mathbf{W}_k + \sum_{l=1}^L \mathbf{R}_l,
\end{equation}  
where \(\mathbf{W}_k = \mathbf{w}_k \mathbf{w}_k^H\) and \(\mathbf{R}_l = \mathbf{r}_l \mathbf{r}_l^H\). Here, \(\mathbf{w}_k \in \mathbb{C}^{N \times 1}\) and \(\mathbf{r}_l \in \mathbb{C}^{N \times 1}\) represent the transmit beamforming vectors for the \(k\)-th medical device and the \(l\)-th patient, respectively.

\subsection{Communication Model}

Once the joint signal is transmitted by the service robot, the received signal at the $k$-th medical device can be expressed as \cite{shikh2007joint}:
\begin{equation}\label{eq1}
    y_k(t) = \mathbf{h}_k^H \mathbf{x}(t) + n_k,
\end{equation}
where $\mathbf{h}_k\in\mathbb{C}^{N\times1}$ is the channel vector for the $k$-th medical device, $\mathbf{x}(t)\in\mathbb{C}^{N\times1}$ is the transmitted signal, and $n_k \sim \mathcal{CN}(0,\sigma^2_c)$ is the communication noise for the $k$-th medical device. On the patient side, the received signal can be formulated by
\begin{equation}\label{eq2}
    y_l(t) = \alpha_l \mathbf{a}^H \left(\theta_l\right) \mathbf{x}(t) + n_l = \mathbf{h}_l^H \mathbf{x}(t) + n_l,
\end{equation}
where $\alpha_l$ is the path-loss coefficient for the $l$-th patient, $n_l \sim \mathcal{CN} \left(0,\sigma^2_r \right)$ is the sensing noise for the $l$-th patient. The steering vector is denoted by $\mathbf{a}\left(\theta_l\right) \in \mathbb{C}^{N\times1}$ with $\theta_l$ being the angle of the chest of $l$-th patient.

With \eqref{eq1}, we can derive the signal-to-interference-plus-noise ratio (SINR) for the $k$-th medical device. It is expressed as:
\begin{equation}\label{eq12}
\gamma_k = \frac{\left|\mathbf{h}_k^H \mathbf{w}_k\right|^2}{\left|\mathbf{h}_k^H \sum_{k'=1, k' \neq k}^K \mathbf{w}_{k'}\right|^2 + \left|\mathbf{h}_k^H \sum_{l = 1}^L \mathbf{r}_{l}\right|^2 + \sigma^2_c}.
\end{equation}

Similarly, with \eqref{eq2}, the SINR for the $l$-th patient is given by
\begin{equation}\label{eq15}
    \Gamma_{l | k}  = \frac{\left|\mathbf{h}_l^H \mathbf{w}_k\right|^2}{\left|\mathbf{h}_l^H \sum_{k'=1, k' \neq k}^K \mathbf{w}_{k'}\right|^2 + \left|\mathbf{h}_l^H \sum_{l = 1}^L \mathbf{r}_{l}\right|^2 + \sigma^2_r}.
\end{equation}

\subsection{Sensing Model}

The echo signal received by the service robot of the \( l \)-th patient is given by:  
\begin{equation}\label{eq4}
    \mathbf{\hat{y}}_l(t) = \beta_l e^{j \phi_l(t - \tau)} \mathbf{a}(\theta_l) \mathbf{a}^H(\theta_l) \mathbf{x}(t - \tau) + \mathbf{n}_l,
\end{equation}  
where $\beta_l$ is the round-trip path loss coefficient and $\mathbf{n}_l \in \mathbb{C}^{N\times1}$ is noise vector with zero mean and variance of $\sigma^2_e \mathbf{I}$. In addition, \( \phi_l(t-\tau) \) captures the effect of chest motion together with the propagation delay $\tau$. It can be formulated as:
\begin{equation}\label{phase change}
    \phi_l(t -\tau) = \frac{4 \pi}{\lambda} \left(m_l\left( t - \tau \right) \right),
\end{equation}  
with \( \lambda \) denoting the signal wavelength. The time delay \( \tau \) is determined by the distance between the service robot and the patient, such as $\tau = \frac{2 d}{c}$, with \( d \) being the patient-robot distance, and $c$ representing the speed of light. Additionally, $m_l(t)$ represents the chest motion of patient $l$.

For each patient, the chest motion $m_l(t)$ is primarily attributed to respiration and heartbeat, and can therefore be modelled as \cite{wang2020vimo, liu2023mmrh, liu2024diversity}:
\begin{equation}\label{motion func}
    m_{l}(t) = m_{0,l} + m_{r,l}(t) + m_{h,l}(t),
\end{equation}
where $m_{0,l}$ denotes body motion considered as interference (e.g., rolling in bed, moving in a wheelchair, or walking), $m_{r,l}(t)$ represents displacement due to respiration, and $m_{h,l}(t)$ represents displacement due to heartbeat. The respiration and heartbeat-induced displacements can be modelled by the sinusoidal functions:
\begin{equation}\label{motion func2} m_{r,l}(t) = A_{r,l} \sin(2\pi f_{r,l} t), \quad m_{h,l}(t) = A_{h,l} \sin(2\pi f_{h,l} t), \end{equation}
where \( A_{r,l} \) and \( A_{h, l} \) denote the amplitudes of respiration and heartbeat, respectively, measured in millimetres, and \( f_{r,l} \) and \( f_{h,l} \) represent their corresponding frequencies, measured in Hertz.
  
\begin{remark}
    Equation \eqref{motion func} can be generalised to represent a variety of biological signals, including those related to sleep monitoring and blood pressure.
\end{remark}

\subsection{Estimation of Patient Position and Physiological Parameters}

From the received echo, the service robot estimates the angle, heart rate, and respiration rate of each patient, i.e., the parameters embedded in \eqref{eq4}.
The first step is to estimate the time delay $\tilde{\tau}$, which provides the distance information between the robot and the patient. This is achieved using a matched filter that correlates the received signal with a time-shifted and conjugated version of the transmitted signal. According to \cite{liu2020joint}, the time delay is estimated as:
\begin{equation}\label{matched filter}
    \Tilde{\tau} = \arg \max_{\tau} \left|\int_0^{\Delta T} \mathbf{\hat{y}}_l\left(t\right) \mathbf{x}^*\left(t-\tau\right)  dt \right|^2,
\end{equation}
where \( \mathbf{x}^*\left(t-\tau\right) \) represents the conjugate of the transmitted signal with a time shift \( \tau \). The integration occurs over the observation interval \( [0, \Delta T] \). With the estimated time delay, we calculate the instantaneous phase from the delay-compensated echo signal as follows:
\begin{equation}
    \hat{\phi}_l (t) = \frac{1}{N}  \sum \tan^{-1} \left( \frac{\Im \{\mathbf{\hat{y}}_l(t) \}} {\Re \{\mathbf{\hat{y}}_l(t) \}}\right), 
\end{equation}
where the calculation averages the phase across all channels, improving robustness against noise and ensuring accurate extraction of respiration and heartbeat components.

With the estimated phase shift $\hat{\phi}(t)$, both respiration and heartbeat components are inherently embedded, resulting in a composite signal. To effectively separate these physiological signatures, we employ the variational mode decomposition (VMD) method, which decomposes $\hat{\phi}_l (t)$ into a finite set of intrinsic mode functions (IMFs) \cite{li2024dynamic}. This can be expressed as $\hat{\phi}_l(t) = \sum_{u=1}^{U} \hat{\phi}_{l,u}(t)$, where $U$ denotes the total number of IMFs. Given the distinct spectral characteristics of respiration and heartbeat signals, they are primarily concentrated in different IMFs. Accordingly, we denote the respiration-related and the heartbeat-related components as $\hat{\phi}_{\mathrm{RR}} = \sum_{i \in \mathcal{U}} \hat{\phi}_i(t) $ and $\hat{\phi}_{\mathrm{HR}} = \sum_{j \in \mathcal{U}} \hat{\phi}_j(t)$, respectively, with \(\mathcal{U} = \{1, \dots, U\}\). The summation reflects the fact that physiological signals may be distributed across multiple IMFs. Once the relevant IMFs are identified, we proceed to estimate the respiration and heartbeat rates using a frequency-domain approach based on the Fourier transform (FT). Although more advanced algorithms exist, such as in \cite{shuzan2021novel}, the FT method provides a straightforward and effective way to analyse frequency components. The process involves four key steps: bandpass filtering, Fourier transform, peak detection, and amplitude estimation. We describe each step below.
\begin{enumerate}
    \item To isolate the respiration and heartbeat components from the estimated phase, we apply bandpass filters tailored to their respective frequency ranges. The filters are denoted by BPF($\cdot$).
    \begin{itemize}
        \item A bandpass filter with a frequency range of 0.1–0.5 Hz is applied to isolate the low-frequency component associated with human respiration. This range corresponds to typical breathing rates of approximately 6–30 respirations per minute (RPM).
        \begin{equation}\label{rpm est}
            \hat{m}_{r, l}(t) = \text{BPF} \left( \hat{\phi}_{\mathrm{RR}}(t), \; 0.1, \; 0.5 \right).
        \end{equation}
        
        \item Subsequently, another bandpass filter with a frequency range of 1–2.5 Hz is applied to extract the high-frequency component associated with cardiac activity. This range corresponds to typical human heart rates of approximately 60–150 beats per minute (BPM).
        \begin{equation}\label{bpm est}
            \hat{m}_{h, l}(t) = \text{BPF} \left( \hat{\phi}_{\mathrm{HR}}(t), \;1, \; 2.5 \right).
        \end{equation}
    \end{itemize}

    \item After filtering, we perform an FT on the respiration and heartbeat components to convert the signals from the time domain to the frequency domain. We use $\mathcal{F}[\cdot]$ to denote the FT. The resulting respiration and heartbeat spectra are given by, respectively:
    \begin{equation}
        M_{r, l}(\omega) = \mathcal{F}\left[ \hat{m}_{r, l}(t) \right], \;  M_{h,l}(\omega) = \mathcal{F}\left[\hat{m}_{h, l}(t) \right].
    \end{equation}

    \item We then identify the dominant peaks in the respiration and heartbeat spectra. The respiration rate and heartbeat rate are determined from these peaks using:
    \begin{equation}\label{freq domain}
        \hat{f}_{r,l} = \argmax_{\omega \in (0.1,0.5)} |M_{r,l} (\omega)|, \; \hat{f}_{h,l} = \argmax_{\omega \in (1,2.5)} |M_{h,l} (\omega)|.
    \end{equation}

    Thus, the heartbeat and respiration rates can be estimated as $60 \times \hat{f}_{h,l}$ BPM and $60 \times \hat{f}_{r,l}$ RPM, respectively.

    \item The amplitudes of the respiration and heartbeat components are estimated by identifying the maximum values of their magnitude spectra within the respective frequency bands. The respiration and heartbeat amplitudes are given by:
    \begin{equation}
        \hat{A}_{r,l} = \max_{\omega \in (0.1,0.5)} |M_{r,l}(\omega)|, \; \hat{A}_{h,l} = \max_{\omega \in (1,2.5)} |M_{h,l}(\omega)|.
    \end{equation}

\end{enumerate}

Finally, we can proceed with angle estimation. By applying high-resolution angle estimation techniques, such as multiple signal classification (MUSIC) \cite{zhang2024joint} or estimation of signal parameters via rotational invariance (ESPRIT) \cite{xiang2023esprit}, the angle of arrival (AoA) can be estimated. By leveraging these signal processing techniques, the service robot can accurately determine the patient's location and monitor vital signs in a contactless manner.

\section{Patient Motion Tracking via IMM Filter}

The accuracy of the vital sign estimation as detailed in Section.~II.D is fundamentally dependent on the quality of the received echo signal. To ensure a strong echo signal, the sensing beams must be precisely and continuously aimed at the patient's chest. However, since patients are dynamic, a static beam is insufficient. Therefore, a critical component of the proposed framework is the ability to track and predict patient motion in real-time. This section presents the IMM filter used to accomplish this, which enables the predictive beamforming essential for robust and reliable sensing.

\subsection{State Evolution Model}
To accurately track patient positions, their kinematic behaviour must be properly modelled. In this paper, we assume that during each time slot, a patient can move in one of four possible directions: left, right, up, down, with constant velocity, or remain stationary. As an example, these movements correspond to rolling on the bed (left/right), transitioning to a seated position (up), lying down from a seated position (down), or remaining still (stationary). Alternatively, these movements can be a patient on a wheelchair moving in one of the four directions, or remaining steady. For analytical tractability, we consider horizontal and vertical movements to be parallel to the antenna array, neglecting any minor rotational or out-of-plane components.

For each patient $l$ at time slot $t$, the state model for left(+) or right(-) movement is formulated as follows \cite{liu2020radar}:
\begin{equation}
    \begin{aligned}
        & \theta_{l,t} = \theta_{l,t-1} \pm d^{-1}_{l,t-1} v_{l,t-1} \Delta T \sin\theta_{l,t-1} + u_\theta,\\
        & d_{l,t} = d_{l,t-1} \mp v_{l,t-1} \Delta T \cos\theta_{l,t-1} + u_d,\\
        & v_{l,t} = v_{l,t-1} + u_v,\\
        & \beta_{l,t} = \beta_{l,t-1} \left(1 + d^{-1}_{l,t-1} v_{l,t-1} \Delta T \cos\theta_{l,t-1} \right) + u_\beta, 
    \end{aligned}
\end{equation}
where $\mathbf{q}_{l,t} = [\theta_{l,t}, d_{l,t}, v_{l,t}, \beta_{l,t}]$ represents the state vector (i.e., angle, distance, velocity, and path loss) of patient $l$ at time slot $t$, with $t = 1, \ldots, T$. The duration of each time slot is denoted by $\Delta T$. The term  $\mathbf{u}_t = [u_\theta, u_d, u_v, u_\beta]$ represents the state process noise.

Following a similar derivation procedure, we derive the state model for up (+) and down(-) movement:
\begin{equation}
    \begin{aligned}
        & \theta_{l,t} = \theta_{l,t-1} \pm d^{-1}_{l,t-1} v_{l,t-1} \Delta T \sin \left(\frac{\pi}{2} + \theta_{l,t-1} \right)+ u_\theta,\\
        & d_{l,t} = d_{l,t-1} \mp v_{l,t-1} \Delta T \cos \left(\frac{\pi}{2} + \theta_{l,t-1}\right) + u_d,\\
        & v_{l,t} = v_{l,t-1} + u_v,\\
        & \beta_{l,t} = \beta_{l,t-1} \left(1 + d^{-1}_{l,t-1} v_{l,t-1} \Delta T \cos\left(\frac{\pi}{2} + \theta_{l,t-1} \right) \right) + u_\beta.
    \end{aligned}
\end{equation}

Additionally, if the patient has no movement, we have the following model:
\begin{equation}
    \begin{aligned}
        & \theta_{l,t} = \theta_{l,t-1} + u_\theta,\quad d_{l,t} = d_{l,t-1} + u_d,\\
        & v_{l,t} = v_{l,t-1} + u_v,\quad \beta_{l,t} = \beta_{l,t-1} + u_\beta.
    \end{aligned}
\end{equation}

\begin{remark}
    We assume the patient is confined to a limited area of movement, such as the dimensions of a bed or a designated region within a room. Once the patient reaches the boundary of this region, they must transition to an alternative state model that accounts for restricted mobility and prevents further physical displacement.
\end{remark}

By denoting the measured parameters as $\mathbf{r}_{l,t} = [\hat{\mathbf{z}}_{l,t},\hat{d}_{l,t}, \hat{v}_{l,t}]$ and the measurement noise as $\mathbf{z}_t = [\mathbf{z}_\theta, z_\tau, z_\mu]$, we can summarise the state model and the measurement model as follows:
\begin{equation}\label{eq55}
    \begin{cases}
      \text{State model:} & \mathbf{q}_{l,t} = \mathbf{g}_1\left(\mathbf{q}_{l,t-1}\right) + \mathbf{u}_t,\\
      \text{Measurement model:} & \mathbf{r}_{l,t} = \mathbf{g}_2\left(\mathbf{q}_{l,t}\right) + \mathbf{z}_t,
    \end{cases}
\end{equation}
where $\mathbf{g}_1\left(\cdot\right)$ is the state transition function, and $\mathbf{g}_2\left(\cdot\right)$ is the measurement function. As $\mathbf{u}_t$ and $\mathbf{z}_t$ are noise vectors with zero-mean Gaussian distribution, their covariance matrices can be formulated by
\begin{equation}\label{eq56}
       \mathbf{Q}_1 = \diag \left(\sigma^2_\theta, \sigma^2_d, \sigma^2_v, \sigma_\beta^2 \right),\; \mathbf{Q}_2 = \diag \left(\sigma^2_1 \mathbf{1}, \sigma^2_{2}, \sigma^2_{3} \right),
\end{equation}
where the formulas for calculating $\sigma^2_{1}$, $\sigma^2_{2}$ and $\sigma^2_{3}$ are given in \cite[Eq. (24)]{liu2020radar}.

\subsection{IMM Filtering for Patient Motion Tracking}

While the extended Kalman filter (EKF) is effective in tracking and predicting a patient's state when their motion follows a single, predefined kinematic model, it lacks adaptability when the patient's movement involves transitions between multiple kinematic models \cite{meng2023vehicular,yuan2025predictive}. This limitation arises because the EKF operates under the assumption of a single, continuous motion model, making it suboptimal for scenarios where patients exhibit different movement patterns over time, such as rolling over in bed, sitting up, or adjusting their posture.

To address this limitation, the IMM filter is employed. The IMM filter enhances tracking accuracy by dynamically adapting to model transitions, ensuring robust state estimation even when the patient switches between different motion modes \cite{kirubarajan2003kalman,xu2024isac}. The IMM filter consists of three key steps: model interaction, elementary filtering, and model probability updating. The EKF is utilised in the second step to estimate the state of each motion model, while the first and third steps facilitate the interaction between different models to achieve a more accurate overall estimation. While the detailed steps can be found in \cite{jo2010integration, nadarajah2012imm, farrell2008interacting}, the key steps are outlined below.

\subsubsection{Model Interaction}
To effectively handle uncertainty in patient motion, the state estimates and the covariance matrices from different motion models must be optimally combined. This combination is achieved through mixing weights, which determine how much effect that each model has on the overall state estimate. The mixing weights are derived from two key probabilities: the model probability \(p^{(j)}_{t-1}\), which represents the likelihood that model \(j\) was correct in the previous time step, and the transition probability \(\pi_ {j,i}\), which defines the probability of switching from model \(j\) to model \(i\). The mixing weights \(p_{t-1}^{(i|j)}\) are computed as:
\begin{equation}
    p^{(i|j)}_{t-1} = \frac{\pi_{j,i} p^{(j)}_{t-1}}{\sum_{j=1}^{M} \pi_{j,i} p^{(j)}_{t-1}},
\end{equation}
where \(M\) represents the total number of motion models.

Using these mixing weights, the predicted state and covariance matrix for model \(i\) are given by:
\footnotesize
\begin{equation}
    \Bar{\mathbf{M}}_{t-1}^{(i)} = \sum_{j=1}^M p^{(i|j)}_{t-1} \left[\hat{\mathbf{M}}_{t-1}^{(j)} + \left(\Bar{\mathbf{q}}_{t-1}^{(i)} - \hat{\mathbf{q}}_{t-1}^{(j)} \right) \left(\Bar{\mathbf{q}}_{t-1}^{(i)} - \hat{\mathbf{q}}_{t-1}^{(j)} \right)^T \right],
\end{equation}
\normalsize
where 
\begin{equation}
    \Bar{\mathbf{q}}_{t-1}^{(i)} = \sum_{j=1}^M \hat{\mathbf{q}}_{t-1}^{(j)} p^{(i|j)}_{t-1}.
\end{equation}

\subsubsection{Elementary Filtering}
Once the mixed state estimates are determined, each model independently implements EKF to recursively estimate the patient's state. The EKF involves six key steps \cite{kay1993fundamentals}:
\begin{enumerate}[i).]
    \item State Prediction:
    \begin{equation}\label{EKF1}
        \hat{\mathbf{q}}_{t|t-1}^{(i)} = g_1^{(i)} \left( \hat{\mathbf{q}}_{t-1}^{(i)} \right),
    \end{equation}
    
    \item Linearisation:
    \begin{equation}\label{EKF2}
        \mathbf{G}_{1, t-1}^{(i)} = \frac{\partial g_1^{(i)}}{\partial \mathbf{q}} \bigg|_{\mathbf{q} = \hat{\mathbf{q}}_{t-1}^{(i)}}, \;
        \mathbf{G}_{2, t}^{(i)} = \frac{\partial g_2^{(i)}}{\partial \mathbf{q}} \bigg|_{\mathbf{q} = \hat{\mathbf{q}}_{t|t-1}^{(i)}},
    \end{equation}
    
    \item Covariance Matrix Prediction:
    \begin{equation}\label{EKF3}
        \mathbf{M}^{(i)}_{t|t-1} = \mathbf{G}_{1, t-1}^{(i)}  \Bar{\mathbf{M}}_{t-1}^{(i)}  \mathbf{G}_{1, t-1}^{(i) H} + \mathbf{Q}_1,
    \end{equation}
    
    \item Kalman Gain Calculation:
    \begin{equation}\label{EKF4}
        \mathbf{S}_{t}^{(i)} = \mathbf{G}_{2, t}^{(i)} \mathbf{M}_{t|t-1}^{(i)} \mathbf{G}_{2, t}^{(i), H} + \mathbf{Q}_2,
    \end{equation}
    \begin{equation}
        \mathbf{K}_{t}^{(i)} = \mathbf{M}_{t|t-1}^{(i)} \mathbf{G}_{2, t}^{(i), H} \left(\mathbf{S}_{t}^{(i)} \right)^{-1},
    \end{equation}
    
    \item State Update:
    \begin{equation}\label{EKF5}
        \bar{\mathbf{r}}_{t}^{(i)} = \mathbf{r}_{t}^{(i)} - g_{2}^{(i)} \left( \hat{\mathbf{q}}_{t|t-1}^{(i)} \right), \;\hat{\mathbf{q}}^{(i)}_{t} = \hat{\mathbf{q}}_{t|t-1}^{(i)} + \mathbf{K}_{t}^{(i)} \bar{\mathbf{r}}_{t}^{(i)},
    \end{equation}
    
    \item Covariance Matrix Update:
    \begin{equation}\label{EKF6}
        \hat{\mathbf{M}}_{t}^{(i)} = \left( \mathbf{I} - \mathbf{K}_{t}^{(i)} \mathbf{G}_{2, t}^{(i)}\right) \mathbf{M}^{(i)}_{t|t-1},
    \end{equation}
\end{enumerate}

\subsubsection{Model Probability Update}
The final step of the IMM filter is updating the probability of each model given the most recent measurement. This is computed using the likelihood function:
\begin{equation}
    \Lambda_{t}^{(i)} = \frac{\exp \left( -\frac{1}{2}  \bar{\mathbf{r}}_{t}^{(i) T } \left(\mathbf{S}_{t}^{(i)} \right)^{-1} \bar{\mathbf{r}}_{t}^{(i)} \right)}{ \sqrt{\left |2\pi \mathbf{S}_{t}^{(i)} \right|} },
\end{equation}
\begin{equation}
    p_{t}^{(i)} = \frac{\Lambda_{t}^{(i)}  \sum_{j=1}^M \pi_{j,i} p_{t-1}^{(j)}}{ \sum_{i=1}^M \left( \Lambda_{t}^{(i)}  \sum_{j=1}^M \pi_{j,i} p_{t-1}^{(j)} \right) }.
\end{equation}

As such, the predicted state is given by:
\begin{equation}
    \hat{\mathbf{q}}_t = \sum_{i=1}^M  p_{t}^{(i)} \hat{\mathbf{q}}^{(i)}_t,
\end{equation}
and therefore the MSE matrix can be obtained via
\begin{equation}
    \hat{\mathbf{M}}_{t} = \sum_{i=1}^M p^{(i)}_{t} \left[\hat{\mathbf{M}}_{t}^{(i)} + \left(\hat{\mathbf{q}}_{t}^{(i)} - \hat{\mathbf{q}}_{t}\right) \left(\hat{\mathbf{q}}_{t}^{(i)} - \hat{\mathbf{q}}_{t} \right) ^T \right].
\end{equation}

\section{Performance Indicators}

This section outlines the key performance indicators used to evaluate the system's performance.

\subsection{Semantic Communication}

The semantic transmission rate is defined as the quantity of bits received by the medical device following the extraction of semantic information, and thus, the formulation is expressed as \cite{zhao2025compression}:
\begin{equation}\label{eq11}
    S_k = \frac{\iota}{\rho_k} \log_2 \left(1+\gamma_k \right),
\end{equation}
where the parameter $\rho_k=\frac{\text{len}(c_k)}{\text{len}(m_k)}, 0\leq \rho_k \leq 1$ represents the semantic extraction ratio, and $\text{len}(\cdot)$ denotes the length of the message. Here, $\iota$ is a scalar value converting the word-to-bit ratio. It is important to note that semantic communication in the proposed framework is optional. When $\rho_k = 1$, semantic communication is disabled, and the system defaults to conventional transmission. For example, in critical-care scenarios, this ensures that raw vital sign data (i.e., the original messages) are transmitted continuously without semantic compression, thereby guaranteeing maximum accuracy and reliability. In addition, the lower bound of $\rho_k$, denoted by $\rho_{k,LB}$, has been derived in \cite{yang2024secure, 11031224}:
\begin{equation}\label{eq14}
   \rho_{k,LB} = \frac{1}{1 - \ln Q + \sum_{g}^G w_{g,k} \log p_{g,k} },
\end{equation}
where $Q$ represents the global lower bound of all the individual Bilingual Evaluation Understudy (BLEU) scores. The BLEU score evaluates how closely the reconstructed message, obtained from the received semantic information, matches the original message. Additionally, $w_{g,k}$ denotes the weight assigned to the $g$-grams, where $G$ is the total number of $g$-grams required to represent a sentence. The precision score $p_{g,k}$ quantifies the accuracy of the message recovered by medical device $k$.

To prevent information leakage, we evaluate the worst-case semantic secrecy rate (SSR). A larger SSR means a lower chance of information leakage. In the worst-case scenario, where the patient has an extensive KB similar to that of the robot and the medical devices, i.e., $\rho_{l|k} = \rho_k$, the semantic transmission rate for the $l$-th patient related to the $k$-th medical device is determined as follows:
\begin{equation}\label{eq16}
    S_{l | k} = \frac{\iota}{\rho_k} \log_2 \left(1+\Gamma_{l | k} \right).
\end{equation}

In this way, the worst-case SSR of the $k$-th medical device is formulated by
\begin{equation}\label{eq17}
    SSR_{k} = \max \left( \min_{l \in L} \left[S_k - S_{l | k} \right ], \; 0 \right).
\end{equation}

\subsection{Computing and Power Consumption}

Semantic information extraction from conventional messages relies heavily on advanced machine learning models, such as Transformers, which are computationally intensive. Therefore, the power consumption of computing plays a pivotal role in semantic communication and must be explicitly accounted for in the overall transmission power budget. In \cite{yang2023energy, zhao2025compression}, the authors established and validated a relationship between the semantic extraction ratio $\rho_k$ and the corresponding computing power, showing a positive correlation within a specific range of $\rho_k$. Accordingly, the computing power function was modelled in a piecewise manner. However, identifying the appropriate range and slope of each segment necessitates extensive hyperparameter tuning. To avoid manual hyperparameter selection, we instead adopt a natural logarithmic function to model the computing power:
\begin{equation}\label{eq18}
    P_{\text{comp}} =  \sum_{k=1}^K -F\ln(\rho_k),
\end{equation}
where $F$ is an energy-efficiency coefficient that is dependent on the CPU design.To demonstrate the difference between our function and \cite[Eq. (10)]{zhao2025compression}, we assume $K = 1$ and re-state \cite[Eq. (10)]{zhao2025compression} as: 
\begin{equation*}
    P_{\text{comp}} = C_{1} \rho_k^{-C_2},
\end{equation*}
where $C_{1}$ and $C_{2}$ are positive constants. By setting $C_1, C_2, F$ to their corresponding values, we obtain Figs. \ref{fittings}.
\begin{figure}[!t]
    \centering
    \begin{subfigure}{0.24\textwidth}
        \centering
        \includegraphics[width=0.8\linewidth]{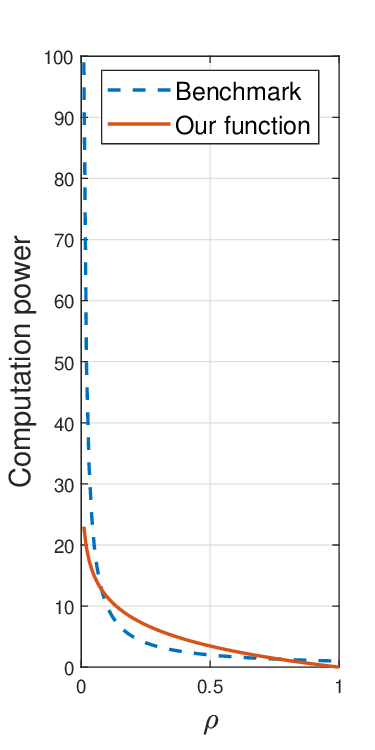}
        \caption{$[1, 1, 5]$}
        \label{fit1}
    \end{subfigure}
    \begin{subfigure}{0.24\textwidth}
        \centering
        \includegraphics[width=0.8\linewidth]{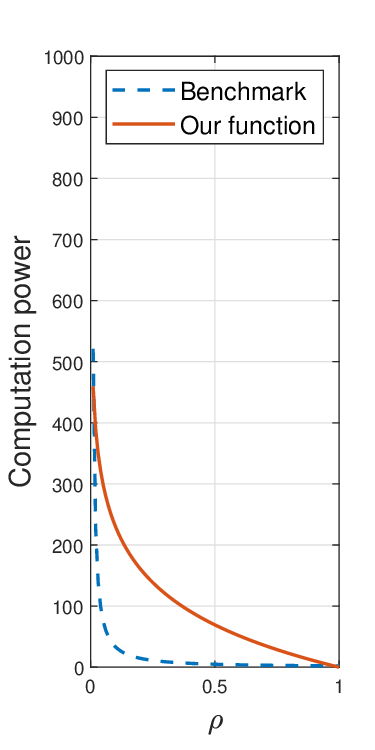}
        \caption{$[2.1, 1.2, 100]$}
        \label{fit2}
    \end{subfigure}
    \caption{An illustration of the approximation of \cite[Eq. (10)]{zhao2025compression} using \eqref{eq18}. The parameters are given as $[C_1, C_2, F]$.}
    \label{fittings}
\end{figure}

The preference for the proposed function is motivated by two key reasons:
\begin{enumerate}[i).]
    \item Eq.~\eqref{eq18} provides a smoother increase of power when $\rho$ decreases from a large value. Whereas \cite[Eq. (10)]{zhao2025compression} only experiences large increments of power when $\rho$ is relatively small.
    \item Eq.~\eqref{eq18} involves only a single hyperparameter, whereas \cite[Eq. (10)]{zhao2025compression} requires two. Moreover, when $\rho = 1$, indicating that no semantic extraction is performed, Eq. \eqref{eq18} yields zero computing power, while \cite[Eq. (10)]{zhao2025compression} produces a small positive value. Since no semantic processing occurs when $\rho = 1$, the computing power should ideally be zero.
\end{enumerate}

On the other hand, the communication and sensing power consumption at the transmitter side is given by \cite{wei2024cramer}:
\begin{equation}\label{eq19}
   P_{\text{c\&s}} =  \Tr \left(\sum_{k=1}^K \mathbf{W}_k + \sum_{l=1}^L \mathbf{R}_l\right).
\end{equation}

Additionally, the overall transmission power consumption is limited to the power budget:
\begin{equation}\label{eq20}
    P_{\text{comp}} + P_{\text{c\&s}} \leq P_t,
\end{equation}
with $P_t$ being the total power budget.

\subsection{Sensing}

To accurately track the motion of the patient and estimate vital sign information from the echo signal, it is essential to precisely direct the beams toward the patient’s chest, corresponding to the optimal measurement model in \eqref{eq55}. This necessitates assessing angle estimation accuracy by calculating and optimising the MSE between the detected and true angles, thereby ensuring that the beams are effectively aligned with the correct directions. However, deriving a closed-form expression for MSE can be challenging. As an alternative, we utilise the Cramér-Rao bound (CRB) to assess the performance of patient sensing in a static scenario. For unbiased estimators, the CRB provides a theoretical lower bound on the MSE, i.e., $\text{MSE}(\hat{\theta}) \geq \text{CRB}(\theta)$, offering an analytically tractable means of performance evaluation. Since the CRB formulation has been derived in \cite{1703855}, we provide a brief description below.

Defining the parameters to be estimated as $\xi_l = [\theta_l, \beta_l]$, the Fisher information matrix (FIM) about $\xi_l$ is given by:
\begin{equation}\label{eq21}
    \mathbf{J}_l = 
    \begin{bmatrix}
        {J}_{\theta_l \theta_l} & \mathbf{J}_{\theta_l \beta_l}\\
        \mathbf{J}_{\theta_l \beta_l}^T & \mathbf{J}_{\beta_l \beta_l}
    \end{bmatrix}.
\end{equation}

Let us define: $\mathbf{B}_l = \mathbf{a}(\theta_l) \mathbf{a}^H(\theta_l)$ and $\Dot{\mathbf{B}}_{\theta_l} = \frac{\partial \mathbf{B}_l}{\partial \theta_l}$. Hence, we can obtain the following equations:
\begin{equation}\label{eq22}
    {J}_{\theta_l \theta_l} = \frac{2T|\beta_l|^2}{\sigma^2_e} \Tr \left(\dot{\mathbf{B}}_{\theta_l} \mathbf{R}_x \dot{\mathbf{B}}_{\theta_l}^H \right), 
\end{equation}
\begin{equation}
    \mathbf{J}_{\beta_l \beta_l} = \frac{2T}{\sigma^2_e} \Tr \left(\mathbf{B}_l \mathbf{R}_x \mathbf{B}_l^H \right) \mathbf{I},
\end{equation}
\begin{equation}\label{eq23}
    \mathbf{J}_{\theta_l \beta_l} = \frac{2T\beta_l^*}{\sigma^2_e} \Re\left\{ \Tr \left(\mathbf{B}_l \mathbf{R}_x \dot{\mathbf{B}}^H_{\theta_l} \right)\right\}[1,\;j].
\end{equation}

Therefore, the CRB of $\theta_l$ is formulated by
\begin{equation}\label{eq25}
    CRB(\theta_l) = {\mathbf{J}_l^{-1}}_{[1,1]}=\left({J}_{\theta_l \theta_l} - \mathbf{J}_{\theta_l \beta_l} \mathbf{J}^{-1}_{\beta_l \beta_l} \mathbf{J}_{\theta_l \beta_l}^T\right)^{-1},
\end{equation}
where $\mathbf{A}_{[1,1]}$ means the element of matrix $\mathbf{A}$ in row 1 and column 1.

\begin{remark}
The advantage of semantic communication over traditional communication methods, while maintaining the sensing performance, will be demonstrated in Section~VI.A.
\end{remark} 

For tracking performance, we focus on the root mean square error (RMSE) of angle and distance, which are given by:
\begin{equation}\label{rmse track}
    \text{RMSE}_{\theta,l} = \sqrt{\left( \theta_{l} - \hat{\theta}_{l} \right)^2}, \;
    \text{RMSE}_{d,l} = \sqrt{\left( d_{l} - \hat{d}_{l} \right)^2},
\end{equation}

Finally, to evaluate the accuracy of vital sign estimation, we adopt the RMSE as a standard performance metric. Specifically, the RMSEs for respiration and heartbeat frequency estimates are defined as
\begin{equation}
\text{RMSE}_{r,l} = \sqrt{\left( f_{r,l} - \hat{f}_{r,l} \right)^2}, \;
\text{RMSE}_{h,l} = \sqrt{\left( f_{h,l} - \hat{f}_{h,l} \right)^2},
\end{equation}

Having defined these key performance indicators for semantic communication, power consumption, and sensing accuracy, we are now equipped to formulate an optimisation problem that balances these competing objectives.

\section{Problem Formulation and Algorithm Design}

This section formulates the optimisation problem and proposes the corresponding algorithms.

\subsection{Problem Formulation}

Since the medical devices are movable within a limited range (e.g., to facilitate better treatment), the communication channels are inherently uncertain. Consequently, the objective is to maximise the worst-case SSR across all medical devices to guarantee robustness. Additionally, to ensure precise beam alignment with the patients’ chest angles, we minimise the sum CRB of the patient angles, as a lower CRB value corresponds to reduced sensing error. This dual optimisation enhances semantic secrecy while improving the precision of patient angle estimates.

The optimisation problem is subject to several constraints. Firstly, the semantic extraction ratio $\rho_k$ is lower-bounded by $\rho_{k,LB}$ as shown in \eqref{eq14}. Furthermore, the constraint of the transmit power budget must be satisfied, as defined in \eqref{eq20}. Additionally, the positive semi-definite property of matrices must be satisfied. Moreover, to facilitate signal transmission through a single-stream transmit beamforming for each medical device, a rank-one constraint is imposed. It is worth noting that the inclusion of this constraint eliminates the need for more complex transceiver schemes, as discussed in \cite{wang2014outage}. Hence, the optimisation problem is formulated as:
\begin{subequations}\label{eq40}
    \begin{align}
        \max_{\mathbf{W}_k, \mathbf{R}_l, \rho_k}\; &  \kappa_1 \min_{k\in K}(SSR_k) - \kappa_2 \sum_{l=1}^L CRB(\theta_l)\label{eq40a}\\
        \text{s.t.} \quad & \rho_{k, LB} \leq \rho_k \leq 1, \forall k,\label{eqGOPb}\\
        & P_{\text{comp}} + P_{\text{c\&s}} \leq P_t, \label{eqGOPd}\\
        &\mathbf{W}_k \succeq 0, \mathbf{W}_k = \mathbf{W}_k^H,  \forall k,\label{eqGOPe}\\
        & \mathbf{R}_l \succeq 0, \mathbf{R}_l = \mathbf{R}_l^H, \forall l,\label{eqGOPf}\\
        & \text{rank}(\mathbf{W}_k) = 1,  \forall k,\label{eqGOPg}
    \end{align}
\end{subequations}
where $\kappa_1$ and $\kappa_2$ are the weights.

\subsection{Algorithm Design}

To simplify the objective function, we first transfer \eqref{eq40} into the following form:
\begin{subequations}\label{eq41}
\begin{align}
    \max_{\mathbf{W}_k, \mathbf{R}_l, \rho_k, \lambda}\; &  \kappa_1 \lambda - \kappa_2 \sum_{l=1}^L CRB(\theta_l)\label{eq41a}\\
    \text{s.t.} \quad & S_k - S_{l | k} \geq \lambda, \forall k, \forall l, \label{eq41b}\\
    & \eqref{eqGOPb} - \eqref{eqGOPg}.
\end{align}
\end{subequations}

Since both medical devices and patients move within a limited range, we assume that the channel uncertainty is confined within a bounded spherical region \cite{wang2014transmit,song2012robust,shenouda2007convex}:
\begin{equation}\label{eq42}
    \mathcal{H}_{i} := \left\{ \left(\hat{\mathbf{h}}_i + \mathbf{u}_i \right)^H \; \Big| \; ||\mathbf{u}_i|| \leq \varepsilon_i \right\}, \;i \in [k,l],
\end{equation}
where $\varepsilon_i \geq 0$ corresponds to the radius of $\mathcal{H}_i$. The estimated channel is denoted by $\hat{\mathbf{h}}_i$, while $\mathbf{u}_i$ represents the error vector limited by $\varepsilon_i$. Next, we deal with the non-concave constraint \eqref{eq41b}. The upper and lower bounds of $|\mathbf{h}_i \sum_{k=1}^K \mathbf{w}_k|^2$ are given by
\begin{equation}\label{eq43}
\begin{aligned}
    & \left|\mathbf{h}_i \sum_{k=1}^K \mathbf{w}_k\right|^2 \leq \hat{\mathbf{h}}_i \sum_{k=1}^K \mathbf{W}_k \hat{\mathbf{h}}_i^H + 2 \varepsilon_i \left|\left| \sum_{k=1}^K \mathbf{W}_k \hat{\mathbf{h}}_i^H\right|\right|,\\
    & \left|\mathbf{h}_i \sum_{k=1}^K \mathbf{w}_k \right|^2 \geq \hat{\mathbf{h}}_i \sum_{k=1}^K \mathbf{W}_k \hat{\mathbf{h}}_i^H - 2 \varepsilon_i \left|\left|\sum_{k=1}^K \mathbf{W}_k \hat{\mathbf{h}}_i^H\right|\right|.
\end{aligned}
\end{equation}

\begin{proof}
    By neglecting the term $\mathbf{u}_i \sum_{k=1}^K \mathbf{w}_k \mathbf{w}_k^H \mathbf{u}_i^H$, we derive the upper bound via
 \begin{equation*}
        \begin{aligned}
            \left|\mathbf{h}_i \sum_{k=1}^K \mathbf{w}_k\right|^2 & = \left(\hat{\mathbf{h}}_i + \mathbf{u}_i \right) \sum_{k=1}^K \mathbf{w}_k \mathbf{w}_k^H \left(\hat{\mathbf{h}}_i + \mathbf{u}_i \right)^H\\
            & \overset{(a)}{\approx} \hat{\mathbf{h}}_i \sum_{k=1}^K \mathbf{w}_k \mathbf{w}_k^H \hat{\mathbf{h}}_i^H + 2 \Re \left\{ \mathbf{u}_i \sum_{k=1}^K \mathbf{w}_k \mathbf{w}_k^H \hat{\mathbf{h}}_i^H \right\}\\
            &\overset{(b)}{\leq} \hat{\mathbf{h}}_i \sum_{k=1}^K\mathbf{W}_k \hat{\mathbf{h}}_i^H  + 2 \left|  \mathbf{u}_i \sum_{k=1}^K \mathbf{W}_k \hat{\mathbf{h}}_i^H \right|\\
            & \overset{(c)}{\leq} \hat{\mathbf{h}}_i \sum_{k=1}^K \mathbf{W}_k \hat{\mathbf{h}}_i^H  + 2 \varepsilon_i \left| \left| \sum_{k=1}^K \mathbf{W}_k \hat{\mathbf{h}}_i^H \right| \right|,
        \end{aligned}
    \end{equation*}
    where step (a) follows from expanding the quadratic form and discarding the error term $\mathbf{u}_i \mathbf{W}_k \mathbf{u}_i^H$. Step (b) uses the inequality $\Re\{z\} \leq |z|$ for any complex number $z$, and step (c) follows from the Cauchy–Schwarz inequality. The proof of the corresponding lower bound follows similarly by using \(\Re\{z\} \geq -|z|\).
\end{proof}

\begin{table*}
\centering
\begin{minipage}{1\textwidth}
    \begin{align} \label{eq44}
      & \frac{\iota}{\rho_k}\log_2 \left(\frac{\overbrace{\hat{\mathbf{h}}_k \sum_{k=1}^K \mathbf{W}_k \hat{\mathbf{h}}_k^H - 2 \varepsilon_k \left|\left|\sum^K_{k=1} \mathbf{W}_k \hat{\mathbf{h}}_k^H\right|\right| + \hat{\mathbf{h}}_k \sum_{l=1}^L \mathbf{R}_l \hat{\mathbf{h}}_k^H - 2 \varepsilon_k \left|\left|\sum^L_{l=1} \mathbf{R}_l \hat{\mathbf{h}}_k^H \right|\right| + \sigma_c^2 }^{\text{Term 1}}}{\underbrace{\hat{\mathbf{h}}_k \sum^K_{k' = 1, k' \neq k} \mathbf{W}_{k'} \hat{\mathbf{h}}_k^H + 2 \varepsilon_k \left|\left|\sum^K_{k'=1, k' \neq k} \mathbf{W}_{k'} \hat{\mathbf{h}}_k^H\right|\right| + \hat{\mathbf{h}}_k \sum^L_{l = 1} \mathbf{R}_{l} \hat{\mathbf{h}}_k^H + 2 \varepsilon_k \left|\left|\sum^L_{l=1} \mathbf{R}_{l} \hat{\mathbf{h}}_k^H\right|\right| + \sigma_c^2}_\text{Term 2} } \right) \nonumber\\
     &- \frac{\iota}{\rho_k}\log_2 \left(\frac{\overbrace{\hat{\mathbf{h}}_l \sum_{k=1}^K \mathbf{W}_k \hat{\mathbf{h}}_l^H + 2 \varepsilon_l \left|\left|\sum_{k=1}^K \mathbf{W}_k \hat{\mathbf{h}}_l^H\right|\right| + \hat{\mathbf{h}}_l \sum^L_{l=1} \mathbf{R}_l \hat{\mathbf{h}}_l^H + 2 \varepsilon_l \left|\left|\sum_{l=1}^L\mathbf{R}_l \hat{\mathbf{h}}_l^H\right|\right| + \sigma_r^2}^{\text{Term 3}} }{\underbrace{\hat{\mathbf{h}}_l \sum_{k'=1, k' \neq k}^K \mathbf{W}_{k'} \hat{\mathbf{h}}_l^H - 2 \varepsilon_l \left|\left| \sum_{k' = 1, k' \neq k}^K \mathbf{W}_{k'} \hat{\mathbf{h}}_l^H\right|\right|+ \hat{\mathbf{h}}_l \sum^L_{l=1} \mathbf{R}_l \hat{\mathbf{h}}_l^H - 2 \varepsilon_l \left|\left|\sum_{l=1}^L\mathbf{R}_l \hat{\mathbf{h}}_l^H\right|\right| + \sigma_r^2}_{\text{Term 4}} } \right) \geq \lambda.
    \end{align}
\medskip
\hrule
\end{minipage}
\end{table*}

In a similar way, the lower and upper bounds of $\left|\mathbf{h}_i \sum \mathbf{r}_l\right|^2$ can be obtained. As such, the lower bound of \eqref{eq41b} is found and shown in \eqref{eq44}. However, \eqref{eq44} remains non-convex, primarily due to the terms $\max -\log_2(\text{Term 2})$ and $\max -\log_2(\text{Term 3})$. Inspired by \cite{zhao2015robust}, to address this issue, we consider the following change of variables in \eqref{eq44}:
\begin{equation}
    \begin{aligned}
        & e^{a_k} \triangleq \text{Term 1}, \quad e^{b_k} \triangleq \text{Term 2},\\
        & e^{c_k} \triangleq \text{Term 3}, \quad e^{d_k} \triangleq \text{Term 4}.
    \end{aligned}
\end{equation}

By introducing exponential terms, the expression such as $\max -\log_2(\text{Term 2})$ is effectively transformed into $\max -b_k \log_2(e)$, which is a linear and convex function. Thus, \eqref{eq44} can be written in the form of:
\begin{equation} \label{eq45}
    \frac{\iota}{\rho_k}\log_2 \left(\frac{e^{a_k}}{e^{b_k}} \right) - \frac{\iota}{\rho_k}\log_2 \left(\frac{e^{c_{l | k}}}{e^{d_{l | k}}} \right) \geq \lambda.
\end{equation}

Finally, to tackle the non-concave function CRB in the objective function, we transform it into a constraint of
\begin{equation}\label{eq35}
    \begin{bmatrix}
        {J}_{\theta_l \theta_l}-{U}_{l} & \mathbf{J}_{\theta_l \beta_l}\\
        \mathbf{J}_{\theta_l \beta_l}^T & \mathbf{J}_{\beta_l \beta_l} 
    \end{bmatrix} \succeq 0, \forall l,\\
\end{equation}
with the CRB part in the objective function becomes $-\kappa_2 \left(\sum_{l=1}^L U_l^{-1}\right)$ and $U_l \geq 0$ is a new variable introduced.

Hence, \eqref{eq41} can be reformulated in the following way:
\begin{subequations}\label{eq46}
    \begin{align}
        \max_{\mathbf{\Psi}}\; &  \kappa_1 \lambda - \kappa_2 \left(\sum_{l=1}^L U_l^{-1} \right)\label{eq46a}\\
        \text{s.t.} \quad & \frac{\iota}{\rho_k}\log_2 e^{a_k - b_k + d_{l | k} - c_{l | k}} \geq \lambda, \forall k, \forall l,\label{eq46b}\\
        & \text{Term 1} \geq e^{a_k},\forall k, \label{eq46c}\\
        & \text{Term 2} \leq e^{b_k}, \forall k, \label{eq46d}\\
        & \text{Term 3} \leq e^{c_{l | k}}, \forall l, \forall k, \label{eq46e}\\
        & \text{Term 4} \geq e^{d_{l | k}}, \forall k, \forall l, \label{eq46f}\\
        & \eqref{eqGOPb},\eqref{eqGOPd},\eqref{eqGOPe}, \eqref{eqGOPf}, \eqref{eqGOPg}, \eqref{eq35}, \label{eq46j}
    \end{align}
\end{subequations}
where $\mathbf{\Psi} = [\mathbf{W}_k, \mathbf{R}_l, U_l, \rho_k, \lambda, a_k, b_k, c_{l | k}, d_{l | k}]$. In addition, it has been verified in \cite{zhao2015robust} that all the inequalities in \eqref{eq46c} to \eqref{eq46f} hold with equalities at the optimal points.

Nevertheless, constraints \eqref{eq46d} and \eqref{eq46e} are still non-convex. We consider using the first-order Taylor Expansion. Let us introduce two more equations:
\begin{equation}\label{eq47}
\begin{aligned}
    \hat{b}_k &= \ln \Biggl(\hat{\mathbf{h}}_k \sum^K_{k' = 1, k' \neq k} \mathbf{W}^i_{k'} \hat{\mathbf{h}}_k^H + 2 \varepsilon_k \left|\left|\sum^K_{k'=1, k' \neq k} \mathbf{W}^i_{k'} \hat{\mathbf{h}}_k^H \right|\right| \\
    & + \hat{\mathbf{h}}_k \sum^L_{l = 1} \mathbf{R}^i_{l} \hat{\mathbf{h}}_k^H + 2 \varepsilon_k \left|\left|\sum^L_{l=1} \mathbf{R}^i_{l} \hat{\mathbf{h}}_k^H\right|\right| + \sigma_c^2 \Biggl),
\end{aligned}
\end{equation}
\begin{equation}\label{eq48}
\begin{aligned}
    \hat{c}_{l | k} &= \ln \Biggl(\hat{\mathbf{h}}_l \sum_{k=1}^K \mathbf{W}^i_k \hat{\mathbf{h}}_l^H + 2 \varepsilon_l \left|\left| \sum_{k=1}^K \mathbf{W}^i_k \hat{\mathbf{h}}_l^H\right|\right|\\
    & + \hat{\mathbf{h}}_l \sum^L_{l=1} \mathbf{R}_l^i \hat{\mathbf{h}}_l^H + 2 \varepsilon_l \left|\left|\sum_{l=1}^L\mathbf{R}_l^i \hat{\mathbf{h}}_l^H\right|\right| + \sigma_r^2 \Biggl).
\end{aligned}
\end{equation}

By applying \eqref{eq47} and \eqref{eq48}, \eqref{eq46d} and \eqref{eq46e} can be replaced by the convex functions \eqref{50} and \eqref{51}, which are shown on the top of the next page.
\begin{table*}
\centering
\begin{minipage}{1\textwidth}
\begin{align}
    &  \hat{\mathbf{h}}_k \sum^K_{k' = 1, k' \neq k} \mathbf{W}_{k'} \hat{\mathbf{h}}_k^H + 2 \varepsilon_k \left|\left|\sum^K_{k'=1, k' \neq k} \mathbf{W}_{k'} \hat{\mathbf{h}}_k^H\right|\right| + \hat{\mathbf{h}}_k \sum^L_{l = 1} \mathbf{R}_{l} \hat{\mathbf{h}}_k^H + 2 \varepsilon_k \left|\left|\sum^L_{l=1} \mathbf{R}_{l} \hat{\mathbf{h}}_k^H\right|\right| + \sigma_c^2 \leq e^{\hat{b}_k} \left(b_k - \hat{b}_k + 1\right), \label{50}\\
    &   \hat{\mathbf{h}}_l \sum_{k=1}^K \mathbf{W}_k \hat{\mathbf{h}}_l^H + 2 \varepsilon_l \left|\left|\sum_{k=1}^K \mathbf{W}_k \hat{\mathbf{h}}_l^H\right|\right| + \hat{\mathbf{h}}_l \sum^L_{l=1} \mathbf{R}_l \hat{\mathbf{h}}_l^H + 2 \varepsilon_l \left|\left|\sum_{l=1}^L\mathbf{R}_l \hat{\mathbf{h}}_l^H\right|\right| + \sigma_r^2 \leq e^{\hat{c}_{l | k}} \left(c_{l | k} - \hat{c}_{l | k} + 1\right) \label{51}. 
\end{align}
\medskip
\hrule
\end{minipage}
\end{table*}
To solve the optimisation problem, we drop the rank-one constraint, then split the problem into three sub-problems and use the alternating optimising method:

\subsubsection{Sub-problem~1}
With given $\rho_k$ in problem \eqref{eq46}, the joint sensing and communication beamforming optimisation problem can be given by
\begin{subequations}\label{eq51}
\begin{align}
    \max_{\mathbf{\Psi}}\; & \kappa_1 \lambda - \kappa_2 \left(\sum_{l=1}^L U_l^{-1} \right)\label{eq51a}\\
    \text{s.t.} \quad & \eqref{eqGOPe}, \eqref{eqGOPf},\eqref{eq35}, \eqref{eq46b},\eqref{eq46c},\eqref{eq46f},\eqref{50},\eqref{51},
\end{align}
\end{subequations}
where $\mathbf{\Psi} = [\mathbf{W}_k, \mathbf{R}_l,U_l, \lambda, a_k, b_k, c_{l | k}, d_{l | k}]$. Optimisation problem \eqref{eq51} is concave and can be effectively solved via the standard convex optimisation tool, such as CVX \cite{grant2014cvx}.

\subsubsection{Sub-problem~2}
When semantic communication is selected (i.e., $\rho_k \neq 1$), given $\mathbf{\Psi} = [\mathbf{W}_k, \mathbf{R}_l, U_l, \lambda, a_k, b_k, c_{l | k}, d_{l | k}]$, as well as denoting $D_k = \log_2 e^{a_k - b_k + d_{l | k} - c_{l | k}}$, the semantic extraction ratio optimisation problem can be given by
\begin{subequations}\label{eq52}
\begin{align}
    \max_{\rho_k}\; & \frac{\iota}{\rho_k} D_k \label{eq52a}\\
    \text{s.t.} \quad & \eqref{eqGOPb}, \eqref{eqGOPd}, \eqref{eq46b},
\end{align}
\end{subequations}

\begin{proposition}
    The optimal value of $\rho_k$, denoted by $\rho_k^*$, is equal to $\min \left( \max \left( \frac{\iota D_k}{\eta^* F}, \rho_{LB, k} \right) , \rho_{UB, k} \right)$, where $\eta$ is the Lagrange multiplier and $\eta^*$ is the optimal value.
\end{proposition}
\begin{proof}
     See Appendix \ref{1Appendix}
\end{proof}

By applying \textbf{Proposition 1}, the optimal value of $\rho_k$ is found.

\subsubsection{Sub-problem~3}
As a last step, we use Gaussian randomisation to recover rank-one solutions.

The detailed steps for solving problem \eqref{eq46} are outlined in Algorithm \ref{alg:1}. The computational complexity of Algorithm \ref{alg:1} is $\mathcal{O}(I_1 I_2 (KN^2)^2) $, where $I_1$ represents the number of iterations in the outer loop, and $I_2$ denotes the number of iterations required to solve the problem \eqref{eq51}. The term $(KN^2)^2$ arises because the problem involves $K$ matrices of size $N \times N$, and the quadratic nature of the optimisation problem results in the squared term.

\begin{algorithm}
\caption{IMM Filter Procedure and Iterative Sensing, Communication, and Semantic Optimisation Algorithm}\label{alg:1}
\begin{algorithmic}[1]
\REPEAT
    \STATE Perform the model interaction step from the IMM Filter.
    \STATE Execute the state prediction, linearization, and MSE prediction steps from the elementary filtering process.
    \STATE Compute the predicted state estimate for predictive beamforming:
    $
        \hat{\mathbf{q}}_{t|t-1} = \sum_{i=1}^M \hat{\mathbf{q}}_{t|t-1}^{(i)} p_{t-1}^{(i)}
    $.
    \STATE Initialize variables: $\rho_k$, $\hat{c}_{l|k, i}$, and $\hat{b}_{k, i}$.

    \REPEAT
        \REPEAT
            \STATE Solve the optimization problem \eqref{eq51} to update $\hat{c}_{l|k, i+1}$ and $\hat{b}_{k, i+1}$.
        \UNTIL{$|\hat{c}_{l|k, i+1} - \hat{c}_{l|k, i}| \leq \varrho_1$ \textbf{and} $|\hat{b}_{k, i+1} - \hat{b}_{k, i}| \leq \varrho_2$.}

        \REPEAT
            \STATE Solve \eqref{eq52} using \textbf{\textit{Proposition 1}} and update $\rho_k$.
        \UNTIL

    \UNTIL{$\left|\mathbf{W}^{i+1} - \mathbf{W}^i\right| \leq \varrho_3$ \textbf{and} $\left|\mathbf{R}^{i+1} - \mathbf{R}^i\right| \leq \varrho_4$.}

    \STATE Apply Gaussian randomisation.
    \STATE Execute the Kalman gain calculation, state update, and MSE matrix update steps from the elementary filtering process.
    \STATE Perform the model probability update step.

\UNTIL

\end{algorithmic}
\end{algorithm}

\section{Numerical Results}

In this section, we present numerical results to evaluate the performance of the proposed framework. The robot is equipped with a ULA consisting of 8 antennas with half-wavelength spacing. The patients are positioned at $(-35^\circ, 2 \text{m})$, $(5^\circ, 5 \text{m})$, and $(25^\circ, 1 \text{m})$, while the medical devices are located at $(-30^\circ, 8 \text{m})$ and $(20^\circ, 4 \text{m})$. The noise power is fixed at $-60$ dBm, and the total transmit power is constrained to 20 dBm. Each channel $\mathbf{h}_i, \; i \in [k,l]$ consists of both line-of-sight (LoS) and NLoS components. In \eqref{eq11}, the parameter $\iota$ is set to 1.1, and the trade-off coefficient $\kappa$ is set to 0.5. The lower bounds of the semantic extraction ratios \(\rho_1\), \(\rho_2\), and \(\rho_3\) are randomly selected based on the dataset\footnotemark[1] \footnotetext[1]{Available at \url{https://www.kaggle.com/datasets/yangtony1999/medical-dataset-for-semantic-communication/data}.}, which has an average value of $0.084$ with standard deviation of $0.089$.
Furthermore, in \eqref{eq42}, the channel uncertainty bound is specified as $\varepsilon_i = 0.01$. The length of the time slot $\Delta T$ is set to 0.1 seconds. For the vital sign modelling, the respiration and heartbeat frequencies are generated in the range $0.1 - 0.5$ Hz and $1-2.5$ Hz, respectively, and amplitudes are generated around $A_r = 5$ mm and $A_h = 1$ mm, respectively. The patients may be moving at a speed of 0.1 m/s. The state transition noise covariance matrix is defined as $\mathbf{Q}_1 = \diag \left(10^{-1}, 1, 10^{-2}, 10^{-4} \right)$.

\subsection{Discussion of Remark 3}

\begin{figure}[!t]
    \centering
    \includegraphics[width=0.8\linewidth]{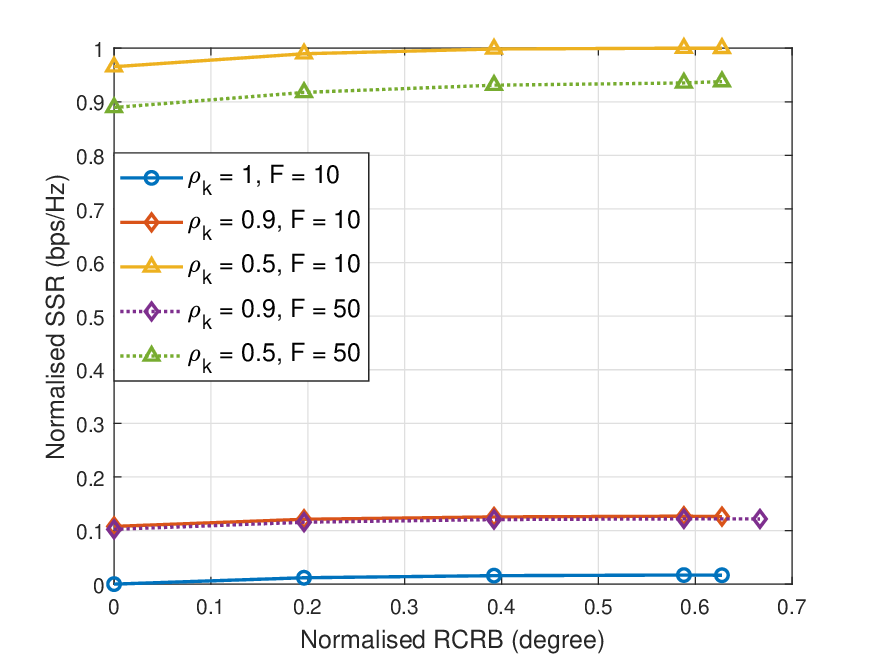}
    \caption{Normalised SSR versus normalised RCRB for different values of $\rho_k$ and $F$. }
    \label{Rho_SSR_CRB}
\end{figure}

For analytical traceability, this subsection considers a simplified scenario involving a single patient and a single medical device. To facilitate a performance comparison between traditional and semantic communication schemes, the parameter $\rho_k$ is set to $1$ to represent traditional communication, while values of $0.9$ and $0.5$ are employed to characterise semantic communication. The value of $F$ is set to $10$ and $50$, illustrating the system performance under varying device computing capabilities. Consequently, sub-problem \eqref{eq52} is omitted from this analysis.

Fig.~\ref{Rho_SSR_CRB} illustrates the normalised SSR and the normalised root CRB (RCRB) with varying values of $\rho_k$ and $F$. The plotted trends demonstrate that, for identical RCRB values, semantic communication yields improved performance over the traditional method, highlighting its ability to enhance communication efficiency without compromising sensing accuracy. Specifically, a lower $\rho_k$ yields a more improved communication performance. Furthermore, as $F$ increases from 10 to 50, a slight degradation in communication performance is observed. This is attributed to the increased power consumption associated with semantic feature extraction at higher $F$, which reduces the available power for communication and sensing, thereby slightly compromising the communication rate to maintain the sensing quality. These findings confirm the advantage of semantic communication in achieving better communication performance while maintaining sensing performance.

\subsection{Semantic Communication Performance}

In this subsection, we compare the performance of the proposed design with the methods presented in \cite{chen2025robust} and \cite{chen2025fast}, respectively. The approach in \cite{chen2025robust} employs the S-procedure for robust beamforming, while \cite{chen2025fast} utilises fractional programming (FP) for non-robust performance optimisation.

\begin{figure}[!t]
\centering
    \begin{subfigure}{0.45\textwidth}
        \centering
        \includegraphics[width=0.9\linewidth]{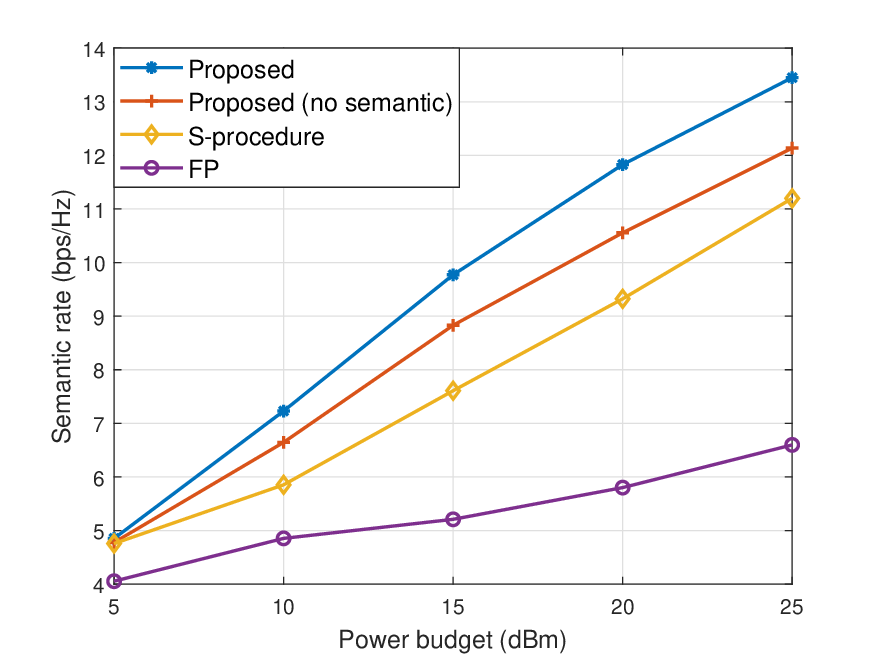}
        \caption{Semantic rate}
        \label{TR}
    \end{subfigure}
    \begin{subfigure}{0.45\textwidth}
        \centering
        \includegraphics[width=0.9\linewidth]{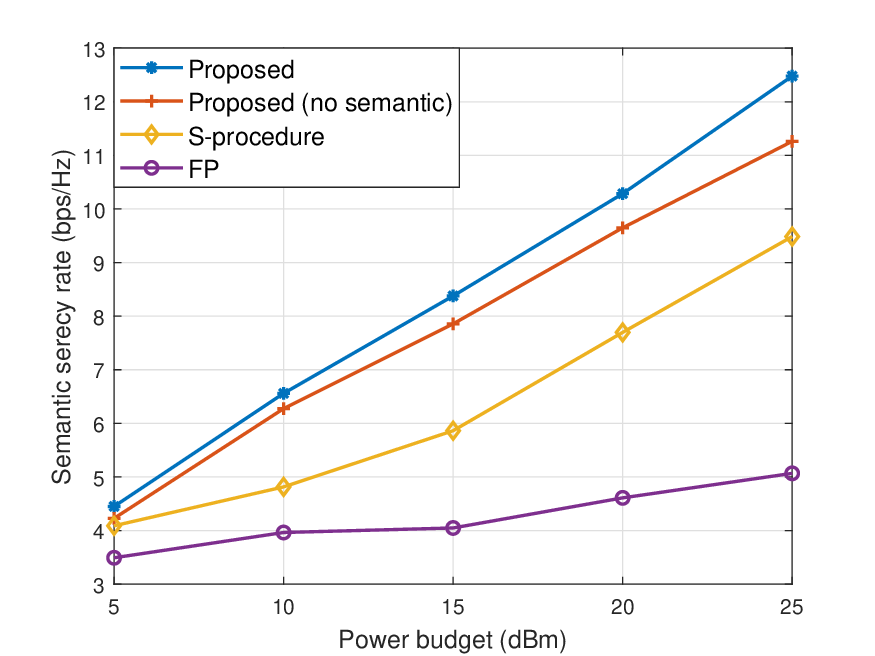}
        \caption{Semantic secrecy rate}
        \label{SSR}
    \end{subfigure}
\caption{Semantic rate and semantic secrecy rate against power budget.}
\label{comm}
\end{figure}

As shown in Fig.~\ref{comm}(a), when the power budget is limited to 5 dBm, the proposed design achieves identical communication performance as the S-procedure-based design. This is primarily because the limited power budget restricts the system's ability to effectively extract semantic features. The non-robust FP method yields the lowest transmission rate due to its sensitivity to channel uncertainties. As the power budget increases, the advantages of both semantic communication and the proposed algorithm become more evident. For instance, at 20~dBm, the proposed design without semantic achieves a transmission rate approximately 10\% higher than the S-procedure approach, and nearly double that of the FP method. When semantic communication is enabled, the performance gain becomes even more significant, with a 20\% improvement over the S-procedure design. Notably, the proposed design maintains consistently high performance as the power budget increases, whereas the FP method shows limited improvement under higher power levels. These findings confirm the advantage of semantic communication in achieving efficient information transmission while maintaining sensing performance.

Fig.~\ref{comm}(b) depicts the SSR as a function of the power budget. The proposed scheme without semantic communication consistently outperforms the S-procedure and FP-based designs, with the performance gap widening as the power budget increases. Incorporating semantic communication further enhances performance, owing to the effective integration of semantic techniques that improve both the semantic rate and the SSR. Notably, the advantage of the proposed design becomes more pronounced at higher power levels.

\begin{figure}[!t]
    \centering
    \includegraphics[width=0.8\linewidth]{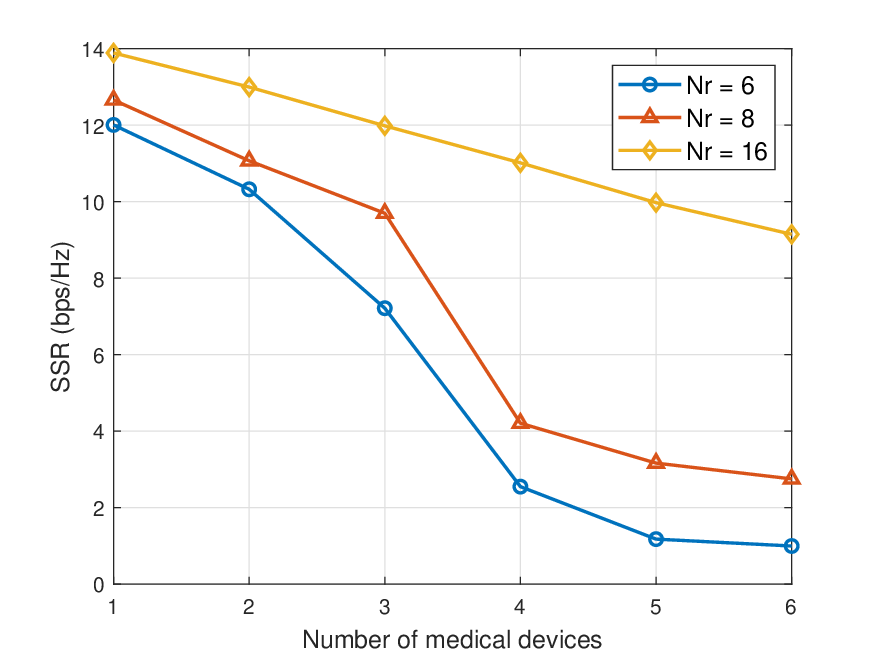}
    \caption{Semantic secrecy rate against the number of medical devices.}
    \label{SSR_k}
\end{figure}

Fig.~\ref{SSR_k} illustrates the average SSR under varying numbers of medical devices and different numbers of antennas. To ensure reliable communication and sensing performance, a common design constraint is that the total number of objects should remain significantly smaller than the number of antennas, i.e., $K + L \ll N$. With the number of patients fixed at $L = 3$, we vary the number of medical devices $K$. When $N = 6$, a sharp decline in SSR is observed from $K = 2$, indicating a violation of the $K + L \ll N$ condition. Increasing the number of antennas to $N = 8$ slightly improves performance, although a noticeable drop still occurs from $K = 3$. With $N = 16$, the SSR exhibits a gradual decline as $K$ increases. This trend arises because adding more devices reduces the per-device resource availability, thereby slightly degrading the SSR while still maintaining performance at high levels.

\begin{figure}[!t]
    \centering
    \includegraphics[width=0.8\linewidth]{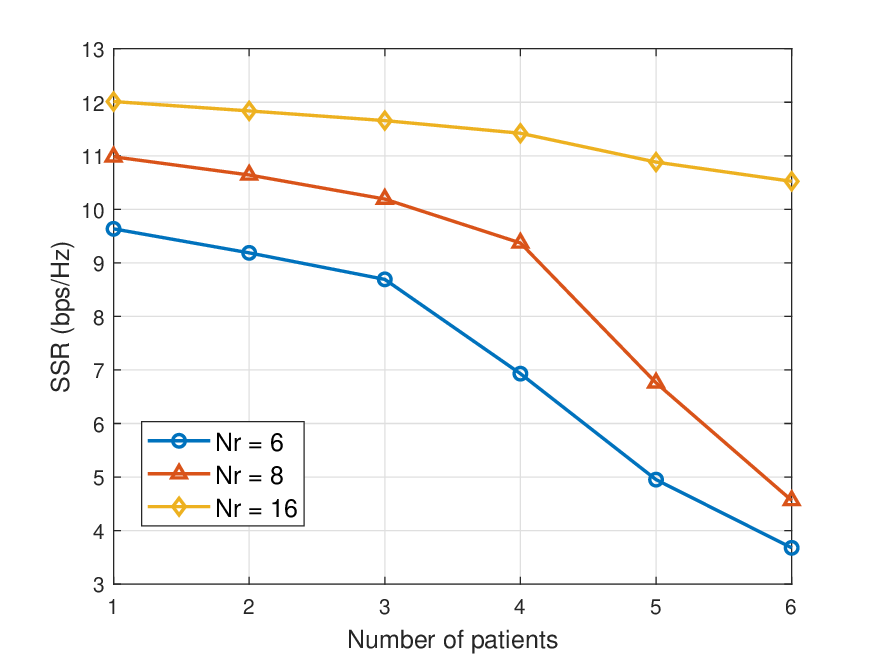}
    \caption{Semantic secrecy rate against the number of patients.}
    \label{SSR_L}
\end{figure}

Fig.~\ref{SSR_L} illustrates the average SSR for varying numbers of patients and antennas. As the number of patients increases, the SSR decreases due to the elevated risk of eavesdropping. When the condition $K + L \ll N$ is satisfied, as in the case of $N = 16$, the SSR remains relatively high. In contrast, when this condition is not met, as with $N = 6$, the SSR experiences a pronounced degradation. Consequently, Figs.~\ref{SSR_k} and \ref{SSR_L} indicate that equipping the robot with a sufficiently large number of antennas is essential to maintain communication and security performance; alternatively, a cooperative robot system could be deployed to achieve the same objective.

\subsection{Sensing and Tracking Performance}

\begin{figure}[!t]
    \centering
    \includegraphics[width=0.8\linewidth]{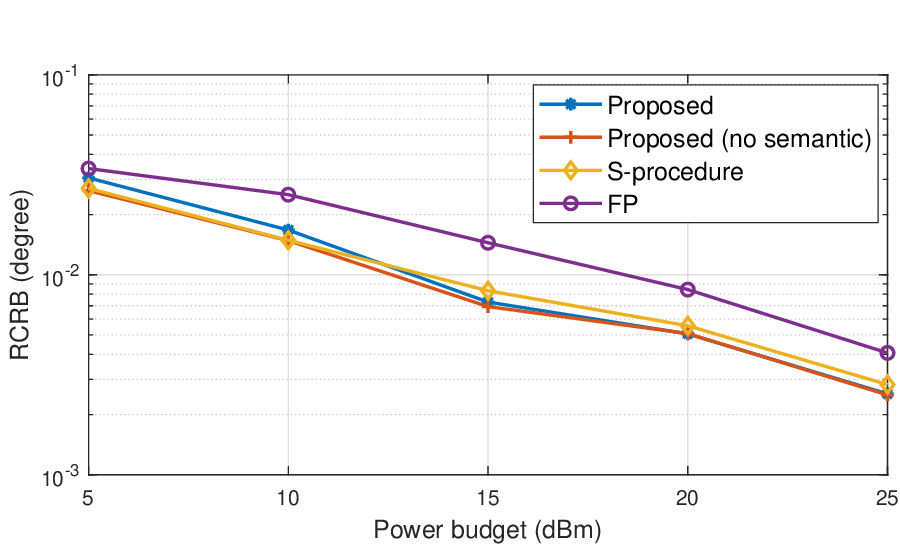}
    \caption{Sum RCRB against power budget.}
    \label{RCRB}
\end{figure}

\begin{figure}[!t]
    \centering
    \includegraphics[width=0.8\linewidth]{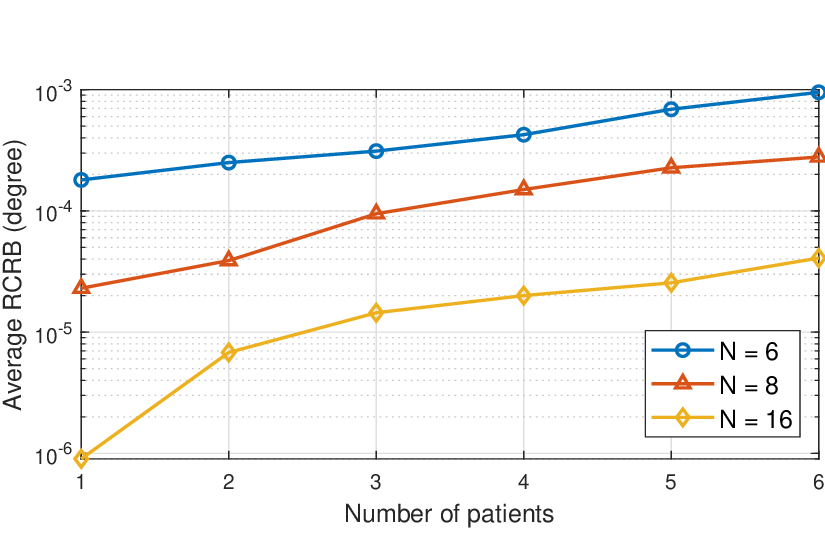}
    \caption{Average sum RCRB against number of patients.}
    \label{RCRB L}
\end{figure}

Regarding sensing performance, Fig.~\ref{RCRB} depicts the sum RCRB values. At lower power levels, the proposed design using semantic communication yields slightly higher RCRB values than the non-semantic scheme, indicating a minor degradation in sensing accuracy due to resource sharing with semantic extraction. However, as the power budget increases, this performance gap gradually diminishes. Notably, at 15 dBm, the proposed design with and without semantic communication achieves nearly equivalent sensing performance and outperforms the S-procedure method. This behaviour confirms that the proposed framework can sustain competitive sensing accuracy under moderate-to-high power budgets while simultaneously delivering enhanced communication performance, as shown in Fig.~\ref{comm}.

Fig.~\ref{RCRB L} illustrates the average sum RCRB values as the number of patients increases. As expected, the sensing accuracy generally deteriorates with an increasing number of patients, primarily due to the resources being divided among more targets. This resource constraint limits the achievable resolution and detection capability for each individual patient. To address this challenge and improve sensing performance, two main approaches can be considered. First, the robot can be equipped with a larger antenna array. For example, increasing the number of antennas from 6 to 16 reduces the RCRB from $10^{-3}$ to $10^{-6}$, thereby enhancing sensing accuracy. Alternatively, deploying a network of cooperating robots enables distributed sensing, where multiple robots jointly process sensing data to exploit spatial diversity and cooperative gains. This observation aligns with the conclusion in Figs.~\ref{SSR_k} and \ref{SSR_L}. Such a cooperative design can effectively mitigate resource constraints, resulting in more robust communication performance and improved sensing accuracy in multi-patient and multi-device scenarios.

Table~\ref{Table I} summarises the average tracking performance over time, quantified by RMSEs shown in \eqref{rmse track}, under varying noise levels, patient speeds, and movement probabilities. As an example, a movement probability of 0.15 indicates that each of the four possible movement directions occurs with probability 0.15, while the probability of no movement is 0.4. When the noise variance $\sigma_e^2 = 1$, corresponding to relatively high noise conditions, the RMSE values are on the order of $10^{-1}$. As the movement probability increases, $\text{RMSE}_{\theta, l}$ exhibits only a marginal rise, demonstrating the effective design of the IMM filter and the beamforming matrices. While $\text{RMSE}_{d,l}$ increases slightly more than $\text{RMSE}_{\theta,l}$ as patient speed increases, it remains within acceptable limits. Notably, increasing the patient’s speed does not lead to higher RMSEs, which highlights the robustness of the IMM filter. Under lower noise conditions, the RMSEs decrease to the order of $10^{-2}$. A similar conclusion can be found, that is, RMSEs remain stable with increasing speed but increase slightly as the movement probability rises. This slight increase in RMSEs with movement probability can be attributed to the challenge of accurately tracking more frequent movements, especially rapid ones. Nevertheless, the IMM filter maintains robust tracking performance under these dynamic conditions.

\begin{table*}[t]
    \centering
    \caption{Average tracking performance across time for varying noise, speed, and patient movement probabilities.}
    \label{Table I}
        \begin{tabular}{ |c|c|c|c|c|}
        \hline
        Noise $\sigma^2_e$ & Speed (m/s) & Probability of each movement & $\text{RMSE}_{\theta,l}$ (degree)  & $\text{RMSE}_{d,l}$ (m) \\
        \hline
        \hline
        & & 0.05 & 0.26843 & 0.18046 \\
        $10^0$ & 0.1 & 0.15 & 0.26856 & 0.18434 \\
        & & 0.25 & 0.26968 & 0.18737 \\
        \hline
        & & 0.05 & 0.26843 & 0.18689 \\
        & 1 & 0.15 & 0.26856 & 0.19295 \\
        & & 0.25 & 0.26969 & 0.19724 \\
        \hline
        & & 0.05 & 0.26843 & 0.22393 \\
        & 3 & 0.15 & 0.26859 & 0.23551 \\
        & & 0.25 & 0.26973 & 0.24438 \\
        \hline
        & & 0.05 & 0.011145 & 0.011112 \\
        $10^{-3}$ & 0.1 & 0.15 & 0.011151 & 0.011118 \\
        & & 0.25 & 0.011175 & 0.011122 \\
        \hline
        & & 0.05 & 0.011145 & 0.011118 \\
        & 1 & 0.15 & 0.011151 & 0.011126 \\
        & & 0.25 & 0.011175 & 0.011132 \\
        \hline
        & & 0.05 & 0.011145 & 0.011141 \\
        & 3 & 0.15 & 0.011151 & 0.011142 \\
        & & 0.25 & 0.011175 & 0.011149 \\
        \hline
        \end{tabular}
\end{table*}

\subsection{Vital Sign Estimation Performance}

\begin{figure*}[t!]
    \centering
    \begin{subfigure}{0.48\textwidth}
        \centering
        \includegraphics[width = \textwidth]{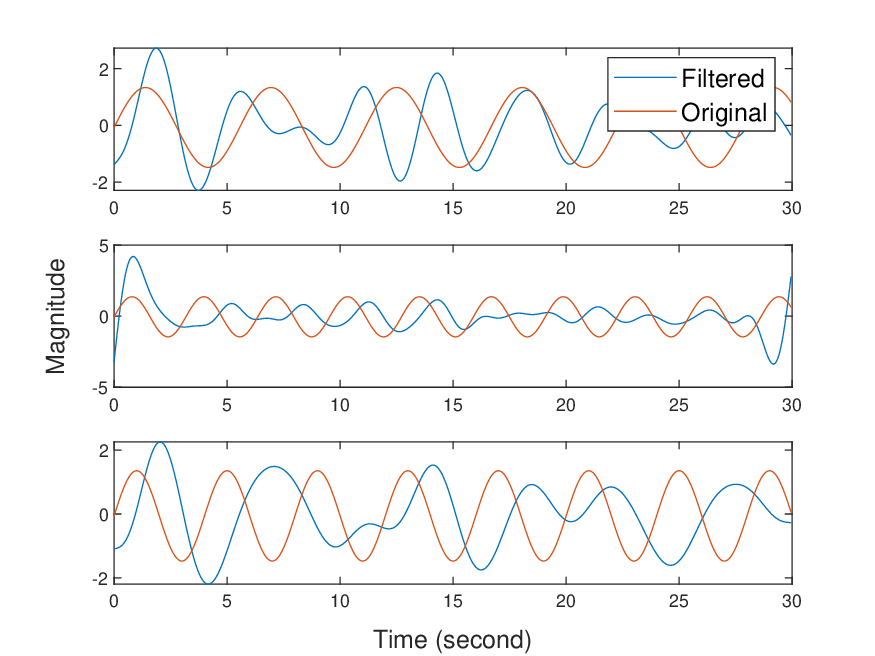}
        \caption{Respiration signal phase over time for patients 1–3.}
    \end{subfigure}
    \begin{subfigure}{0.48\textwidth}
        \centering
        \includegraphics[width = \textwidth]{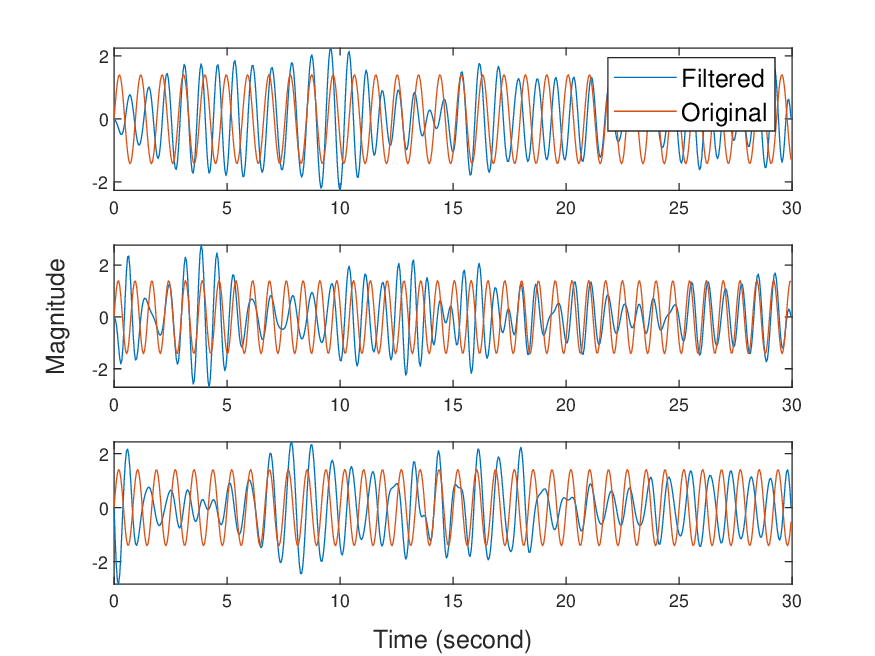}
        \caption{Heartbeat signal phase over time for patients 1–3.}
    \end{subfigure}
    \caption{Temporal signal phase variations.}
    \label{Temporal signal}
\end{figure*}

\begin{figure*}[!t]
    \centering
    \begin{subfigure}{0.48\textwidth}
        \centering
        \includegraphics[width = \textwidth]{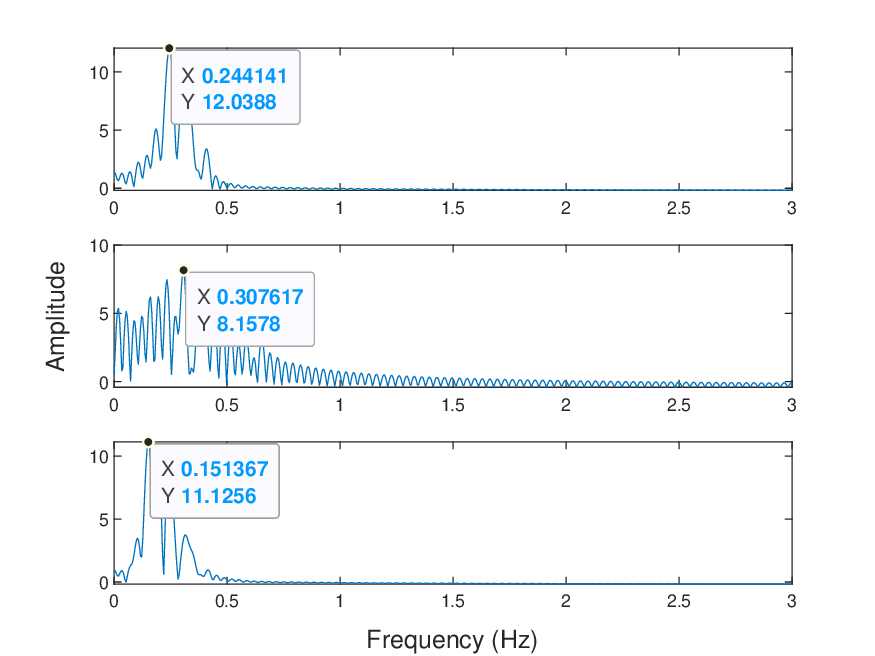}
        \caption{Frequency spectrum of respiration signals for patients 1–3.}
    \end{subfigure}
    \begin{subfigure}{0.48\textwidth}
        \centering
        \includegraphics[width = \textwidth]{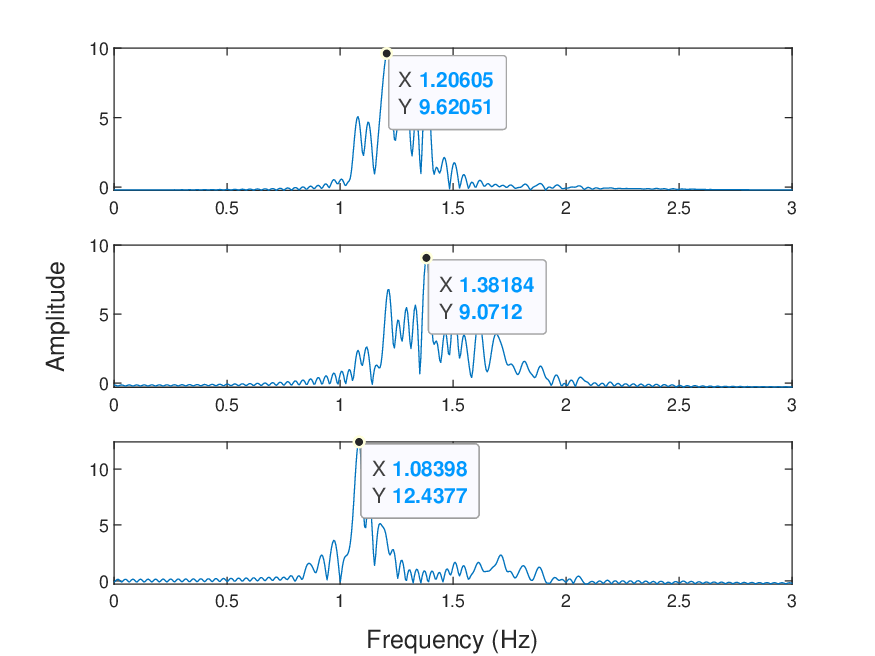}
        \caption{Frequency spectrum of heartbeat signals for patients 1–3.}
    \end{subfigure}
    \caption{Frequency-domain analysis.}
    \label{freq ana}
\end{figure*}

Fig.~\ref{Temporal signal} presents the experimental results of multi-patient vital sign sensing. Specifically, Fig.~\ref{Temporal signal}(a) illustrates the time-domain respiration signals, while Fig.~\ref{Temporal signal}(b) depicts the time-domain heartbeat signals. These signals are obtained after applying bandpass filtering as described in \eqref{rpm est} and \eqref{bpm est}, respectively. The corresponding ground truth values are provided in \eqref{motion func2}. As shown, the filtered signals retain the periodic characteristics of the original waveforms, thereby validating the effectiveness of the proposed design. The preservation of physiological periodicity demonstrates the capability of the system to extract vital sign information, even in a multi-subject sensing scenario.

Fig.~\ref{freq ana} illustrates the frequency-domain spectra (i.e., \eqref{freq domain}) of the extracted physiological signals for each of the three patients. Specifically, Fig.~\ref{freq ana}(a) presents the spectral representation of the respiration signals. The ground truth respiration frequencies are $f_{r,1} = 0.25$ Hz, $f_{r,2} = 0.3145$ Hz, and $f_{r,3} = 0.18$ Hz. As observed, the estimated spectral peaks closely match the true respiration frequencies across all patients, confirming the accuracy of the proposed method. Similarly, Fig.~\ref{freq ana}(b) shows the spectra of the heartbeat signals, with corresponding ground truth frequencies of $f_{h,1} = 1.2$ Hz, $f_{h,2} = 1.345$ Hz, and $f_{h,3} = 1.0578$ Hz. The observed spectral peaks align well with these values, further validating the reliability of the proposed design in capturing vital sign information.

\begin{figure}[!t]
    \centering
    \includegraphics[width=0.8\linewidth]{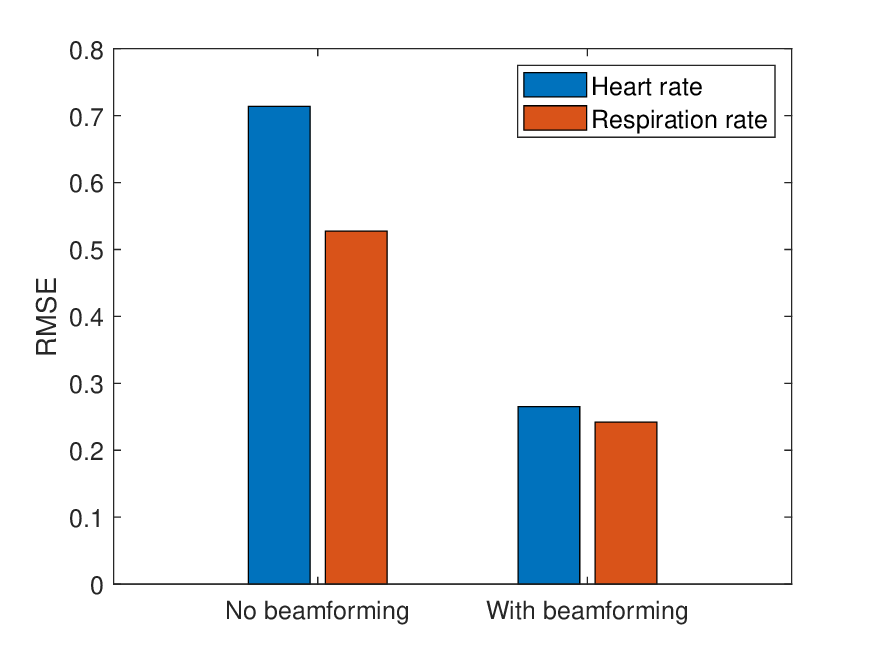}
    \caption{Effect of beamforming on heartbeat and respiration rate estimation.}
    \label{vital sign error}
\end{figure}

Fig.~\ref{vital sign error} evaluates the impact of beamforming on the accuracy of heart rate and respiration rate estimations. The results clearly demonstrate that beamforming significantly reduces estimation error for both physiological parameters. Specifically, for heart rate estimation, beamforming leads to a nearly 60\% reduction in RMSE. A comparable improvement is observed in respiration rate estimation, where beamforming reduces the RMSE by approximately 55\%. These findings underscore the effectiveness of beamforming in enhancing the precision of vital sign monitoring.

\begin{table*}[!t]
    \centering
    \caption{Heart rate and respiration rate estimation error with single patient and multiple patients.}
    \label{Table II}
        \begin{tabular}{ |c|c|c|c|c| }
        \hline
         & Noise $\sigma_e^2$ & Patient Number & $\text{RMSE}_{h,l}$ (BPM) & $\text{RMSE}_{r,l}$ (RPM)\\
        \hline
        \hline
        Single Patient (Joint design) & $10^{-3}$ & Patient 1 &  0.538553 & 0.774975 \\
        \hline
        Single Patient (Sensing-only design) & $10^{-3}$ & Patient 1 &  0.386541 & 0.502827 \\
        \hline
        & & Patient 1 & 1.344375 & 0.968246 \\
        Multiple Patients (Using VMD \& IMFs) & $10^0$ & Patient 2 & 1.486659 & 0.642626 \\
        & & Patient 3 & 1.253420 & 1.445899  \\
         \hline
         & & Patient 1 & 0.810093 & 0.968246 \\
        Multiple Patients (Using VMD \& IMFs) & $10^{-3}$ & Patient 2 & 1.486659 & 0.642626 \\
        & & Patient 3 & 1.130528 & 1.310713 \\
         \hline
        & & Patient 1 & 3.303053 & 2.432688 \\
        Multiple Patients (FFT only) & $10^0 $& Patient 2 & 1.486659 & 3.128573 \\
        & & Patient 3 & 4.409326 & 2.035045\\
        \hline
        & & Patient 1 & 3.303053 & 2.432688 \\
        Multiple Patients (FFT only) & $10^{-3} $& Patient 2 & 1.486659 & 3.128573 \\
        & & Patient 3 & 3.031213 & 1.310713 \\
         \hline
        \end{tabular}
\end{table*}

Table~\ref{Table II} compares the HR and RR estimation performance for single and multi-patient scenarios, evaluated using RMSE. In the single-patient case, Patient 1’s ground truth values are 72 BPM and 15 RPM. The benchmark sensing-only design achieves RMSEs of 0.3866 BPM and 0.5028 RPM. The proposed joint design achieves an RMSE of 0.5386 BPM and 0.7749 RPM. While the RR error increases by approximately 54\% relative to the sensing-only design, the relative errors remain low (0.75\% for HR and 5.17\% for RR), indicating that the joint design preserves strong estimation accuracy. In the multi-patient scenario, inter-user interference leads to performance degradation. Using VMD combined with IMFs to separate composite signals, Patient 1 achieves HR and RR RMSEs of 0.8101 BPM and 0.9682 RPM, corresponding to relative errors of 1.13\% and 6.45\%, respectively. Patient 2, with ground truth values of 80.7 BPM and 18.87 RPM, obtains RMSEs of 1.4867 BPM and 0.6426 RPM (relative errors of 1.84\% and 3.41\%). Patient 3’s HR and RR are 63.47 BPM and 10.80 RPM, with RMSEs of 1.1305 BPM and 1.3107 RPM (relative errors of 1.97\% and 12.14\%). Compared to the FFT-only method, the VMD \& IMFs approach delivers substantial accuracy gains across all patients, highlighting its effectiveness in mitigating multi-user interference. Using the VMD \& IMFs method, when the noise variance increases from $10^{-3}$ to $10^0$, a significant degradation in HR estimation is observed, while RR estimation experiences a smaller performance drop. This difference arises because the amplitude of HR signals is much smaller than that of RR signals, making HR more susceptible to environmental noise. In contrast, the FFT-based method is less sensitive to noise variations but yields substantially lower estimation accuracy overall.

\section{Conclusion and Future Directions}

In this paper, we have proposed an integrated framework that combines sensing, computing, and semantic communication for e-health applications, with a particular focus on vital sign monitoring. The system comprises a service robot, multiple healthcare devices, and several patients. The service robot employs IMM filters to track patient movements and design predictive sensing beamformers accordingly. To enhance data efficiency and protect privacy, we have applied semantic extraction to select representative information from the sensing data. We have evaluated the system performance using key metrics such as semantic computing power and the CRB. We have formulated a joint optimisation problem involving beamforming and semantic extraction ratio, addressing its non-convex nature through a combination of bounding techniques and the bisection method. Simulation results have demonstrated that the proposed framework and algorithms outperform conventional joint sensing and communication methods by achieving higher sensing accuracy, improved semantic transmission efficiency, and enhanced privacy preservation.

As future directions, hardware implementation of the proposed system represents a promising path for practical eHealth applications. Incorporating machine learning techniques can further enhance the robustness of vital sign extraction in noisy and reflective environments. Additionally, integrating complementary sensing modalities, such as cameras, may provide more benefits to system performance. Finally, evaluating the proposed approach within multi-robot or cooperative robotic systems would broaden its applicability and impact in eHealth scenarios.

\appendices

\section{Proof of Proposition 1}
\label{1Appendix}
\setcounter{equation}{0}
\numberwithin{equation}{section}
Through considering the Lagrange multiplier $\eta$, problem \eqref{eq52} can be transferred to 
\begin{subequations}\label{appendixeq1}
\begin{align}
    \max_{\rho_k} &  \sum_{k=1}^K  \frac{\iota}{\rho_k} D_k + \eta (F \sum_{k=1}^K \ln \rho_k - P_{c\&s} + P_t) \label{appendix1a}\\
    \text{s.t.} \; & \; \rho_{LB, k} \leq \rho_k \leq \rho_{UB, k}, \forall k. \label{appendix1b}
\end{align}
\end{subequations}

The maximum value of the function can be found by taking the first-order derivative. Since this is a multi-variable function, the partial derivative of each $\rho_k, k \in K$ is given by:
\begin{equation}\label{appendixeq2}
    \frac{\partial}{\partial \rho_k} = \frac{- \iota D_k}{\rho_k^2} + \frac{\eta F}{\rho_k}.
\end{equation}

Therefore, $\rho_k^* = \min \left( \max \left( \frac{\iota D_k}{\eta F}, \rho_{LB, k} \right) , \rho_{UB, k} \right)$ is found when $\frac{\partial}{\partial \rho_k} = 0$. When the optimal value of $\eta$, denoted by $\eta^*$, is found, the optimal value of $\rho$ is also found. The optimal value of $\eta$, denoted by $\eta^*$, can be found by using the bisection method. 

\bibliographystyle{ieeetr}
\bibliography{bib}

\end{document}